\documentclass[a4paper,10pt]{article}

\usepackage[utf8]{inputenc}
\usepackage[english]{babel}

\usepackage{tikz-cd}

\usepackage{mathtools}
\usepackage{amsmath}
\usepackage{amssymb}
\usepackage{dsfont}
\usepackage{amsthm}
\usepackage{cases}
\usepackage{stix}
\usepackage{pgf,tikz}
\usetikzlibrary{arrows}\usepackage[pdftex,colorlinks]{hyperref}
\hypersetup{                    %
	      allcolors = blue}
	      
\usepackage[textwidth=6in,textheight=9.25in, paperwidth=7.5in,paperheight=11.25in, inner=2.5cm, outer=2.5cm]{geometry} 

\theoremstyle{definition} \newtheorem{defn}{Definition}[section]
\theoremstyle{definition} 
\theoremstyle{remark} 
\theoremstyle{plain} \newtheorem{thm}[defn]{Theorem}
		     \newtheorem*{theorem*}{Theorem}
\theoremstyle{plain} \newtheorem{cor}[defn]{Corollary}
		     
\theoremstyle{plain} \newtheorem{prop}[defn]{Proposition}
\theoremstyle{plain} \newtheorem{prob}[defn]{Problem}
\theoremstyle{plain} \newtheorem{lem}[defn]{Lemma}
\theoremstyle{remark} \newtheorem{rem}[defn]{Remark}
\theoremstyle{plain} 
		      \newtheorem*{rmk*}{Remark}
\theoremstyle{plain} \newtheorem*{princ*}{Principle}

\usepackage{pgf,tikz}
\usepackage{mathrsfs}
\usetikzlibrary{arrows}
\usetikzlibrary[patterns]
\definecolor{ffqqqq}{rgb}{1,0,0}
\definecolor{ududff}{rgb}{0.30196078431372547,0.30196078431372547,1}

\allowdisplaybreaks

\usepackage[misc]{ifsym}

\author{Mark Wilkinson\thanks{\emph{Correspondence Address} (\Letter): New Hall Block, Department of Physics and Mathematics, Clifton Campus, Nottingham Trent University, Nottingham, United Kingdom. \emph{E-mail Address} (@): \href{mark.wilkinson@ntu.ac.uk}{mark.wilkinson@ntu.ac.uk}.}}
\title{A Lie Algebra-Theoretic Approach to Characterisation of Collision Invariants of the Boltzmann Equation for General Convex Particles}
\date{}

\hypersetup{citecolor=red}

\usepackage{mathrsfs}

\usepackage[shortlabels]{enumitem}

\begin{document}
\maketitle
\begin{abstract}
\noindent By studying scattering Lie groups and their associated Lie algebras, we introduce a new method for the characterisation of collision invariants for physical scattering families associated to smooth, convex hard particles in the particular case that the collision invariant is of class $\mathscr{C}^{1}$. This work extends that of Saint-Raymond and Wilkinson (\emph{Communications on Pure and Applied Mathematics} (2018), 71(8), pp. 1494-1534), in which the authors characterise collision invariants only in the case of the so-called \emph{canonical} physical scattering family. Indeed, our method extends to the case of \emph{non-canonical} physical scattering, whose existence was reported in Wilkinson (\emph{Archive for Rational Mechanics and Analysis} (2020), 235(3), pp. 2055-2083). Moreover, our new method improves upon the work in Saint-Raymond and Wilkinson as we place no symmetry hypotheses on the underlying non-spherical particles which make up the gas under consideration. The techniques established in this paper also yield a new proof of the result of Boltzmann for collision invariants of class $\mathscr{C}^{1}$ in the classical case of hard spheres.
\end{abstract}
\section{Introduction}
The characterisation of collision invariants is a fundamental result in the theory of the Boltzmann equation. Knowledge of the structure of all collision invariants allows one to characterise all equilibria of the Boltzmann equation, and facilitates the analysis of both the linear and non-linear Boltzmann equations, amongst other aspects of their analysis. In the case of hard sphere particles, the problem has a long history: we refer the reader to the monograph of Cercignani, Illner and Pulvirenti (\cite{MR1307620}, Chapter 3.1) for details thereon. However, the study of the case of hard \emph{non-spherical} particles is much more recent and less well understood. We recall that a Lebesgue-measurable map $\phi:\mathbb{R}^{3}\rightarrow\mathbb{R}$ is said to be a \textbf{collision invariant} for physical hard sphere scattering if and only if it satisfies the functional equation
\begin{equation*}
\phi(v_{n}')+\phi(\overline{v}_{n}')=\phi(v)+\phi(\overline{v})    
\end{equation*}
for $\mathscr{L}^{3}$-almost every $v, \overline{v}\in\mathbb{R}^{3}$ and $\mathscr{H}$-a.e. $n\in\mathbb{S}^{2}$, where $v_{n}', \overline{v}_{n}'\in\mathbb{R}^{3}$ denote the $n$-dependent binary scattering velocity variables defined by
\begin{equation}\label{spherescat}
\left\{
\begin{array}{l}
v_{n}':=v-((v-\overline{v})\cdot n)n, \vspace{2mm}\\
\overline{v}_{n}':=\overline{v}+((v-\overline{v})\cdot n)n,
\end{array}
\right.
\end{equation}
while $\mathscr{L}^{3}$ and $\mathscr{H}$ denote the Lebesgue measure on $\mathbb{R}^{3}$ and the Hausdorff measure on $\mathbb{S}^{2}$, respectively. Without loss of generality, the unbarred velocity variables denote those of a sphere centred at the origin, whilst the barred velocity variables denote those of the other sphere at collision: see Figure \ref{fig:ballscoll} for an illustration. The unit vector $n\in\mathbb{S}^{2}$, which we call the \emph{collision parameter}, represents the direction of the vector connecting the centres of mass of the two congruent hard spheres at collision, specifically directed from the centre of the unbarred particle to that of the barred. It is well known that a collision invariant is necessarily of the shape
\begin{equation*}
\phi(v)=a+b\cdot v+c|v|^{2}    
\end{equation*}
for $\mathscr{L}^{3}$-a.e. $v\in\mathbb{R}^{3}$, for some constants $a, c\in\mathbb{R}$ and some constant $b\in\mathbb{R}^{3}$. On an intuitive level, a collision invariant is a linear combination of the \emph{mass}, the components of the \emph{linear momentum} and the \emph{kinetic energy} of a given hard sphere. The proof of this result has been achieved under various hypotheses on the regularity of the collision variant $\phi$ by Boltzmann \cite{boltzmann2012wissenschaftliche}, Gr\"{o}nwall \cite{MR1503514}, Cercignani \cite{MR1049048}, and Arkeryd \cite{MR0339666}, among others. A large variety of techniques spans the work of the aforementioned authors. However, techniques from group theory can be employed to help in unifying their efforts. Indeed, as was noted recently in Saint-Raymond and Wilkinson \cite{saint2018collision}, at the heart of the problem of characterisation of collision invariants lies the following problem of group theory:
\begin{prob}
Suppose $q\geq 2$. Given a strict subset $U\subset\mathrm{O}(q)$, find the topological closure of the matrix group generated by $U$.
\end{prob} 
For instance, in \cite{saint2018collision} for the case of hard spheres, $q=3$ and $U\subset\mathrm{O}(3)$ is the set of collision parameter-dependent reflection matrices given by
\begin{equation*}
U:=\{
I-2n\otimes n\in\mathrm{O}(3)\,:\,n\in\mathbb{S}^{2}
\},
\end{equation*}
whose closure in the sense above is the whole orthogonal group $O(3)$. Section \ref{hardsphereproof} below provides detail on this claim. In this article, following the work of \cite{saint2018collision}, we continue the study of collision invariants for the case of convex \emph{non-spherical} particles. We introduce a new method for the characterisation of collision invariants, one which is centred on finding the Lie algebra associated to so-called (minimal) \emph{scattering Lie groups}. In particular, in the new case of what we term \emph{non-canonical} scattering, we show that collision invariants $\varphi$ of class $\mathscr{C}^{1}$ are necessarily of the form
\begin{equation}\label{wordymaintheorem}
\begin{array}{c}
\varphi=\text{function of orientation}+\text{linear function of linear momentum} \vspace{2mm} \\
+\hspace{1mm} \text{linear function of kinetic energy}.
\end{array}
\end{equation}
We invite the reader to consult Theorem \ref{maintheorem} for a precise statement of our main result. Unlike the work in \cite{saint2018collision}, our results in this article do not impose any symmetry constraints on the underlying convex hard particles. Moreover, our results hold for all convex hard particles, not simply those which are strictly convex. In this article, we focus primarily on the case of convex hard particles in 2-dimensions to illustrate our method for the characterisation of collision invariants. We consider the case of \emph{all possible} physical scattering for 2- and 3-dimensional convex hard particles in future work.
\subsection{Non-spherical Particles and Non-uniqueness of Physical Scattering}
The case of physical hard sphere dynamics without rotation is significantly simpler than that of non-spherical particles, as there is only one way by which one can resolve a collision between two hard spheres so that momentum (both linear and angular) and kinetic energy are conserved. In particular, using the terminology of \cite{wilkinson2020non}, there is a \emph{unique physical scattering family} in the case of hard spheres without rotation. Let us review some elementary facts on the case of hard spheres, before we progress to the case of binary non-spherical particle systems. 
\subsubsection{Scattering of Spherical Particles}
\begin{figure}
    \centering
    \includegraphics[scale=0.45]{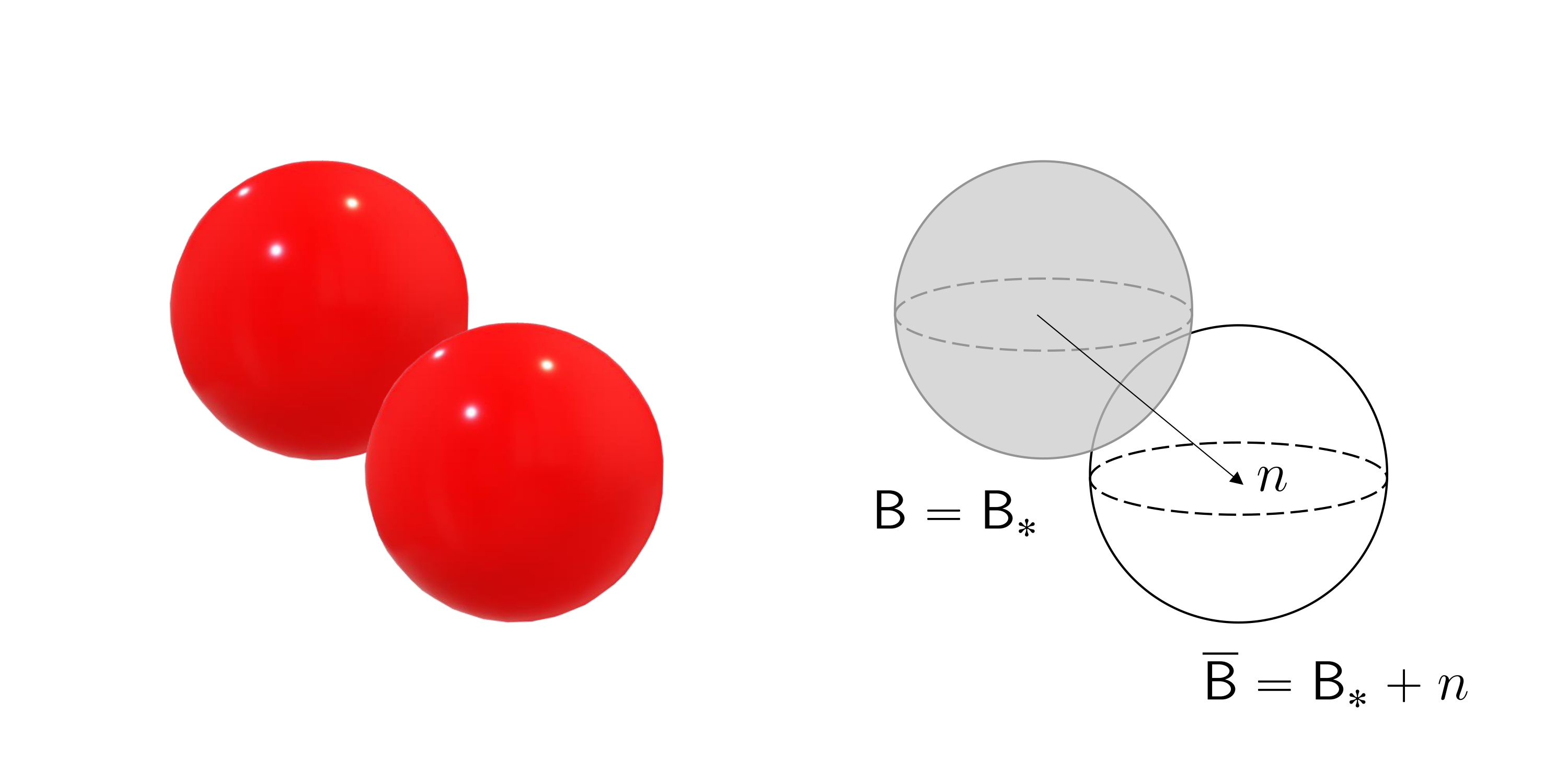}
    \caption{A collision configuration of two hard spheres $\mathsf{B}$ and $\overline{\mathsf{B}}$ in $\mathbb{R}^{3}$, each of which is congruent to a given reference set $\mathsf{B}_{\ast}:=\{y\in\mathbb{R}^{3}\,:\, |y|\leq \frac{1}{2}\}$. The collision parameter $n\in\mathbb{S}^{2}$ represents the direction from the centre of mass of the unbarred sphere to that of the barred.}
    \label{fig:ballscoll}
\end{figure}
We assume without loss of generality that the centre of mass of the sphere represented by unbarred variables lies at the origin in $\mathbb{R}^{3}$ at collision, and also that both spheres are of unit diameter and of unit mass. For a given fixed collision parameter $n\in\mathbb{S}^{2}$, the scattering velocity variables $v_{n}'$ and $\overline{v}_{n}'$ are assumed to obey the conservation of linear momentum, i.e.
\begin{equation}\label{spherecolm}
v_{n}'+\overline{v}_{n}'=v+\overline{v},
\end{equation}
as well as the the conservation of angular momentum (with respect to a point of measurement $a\in\mathbb{R}^{3}$) given by 
\begin{equation}\label{spherecoam}
-a\times v_{n}'+(n-a)\times\overline{v}_{n}'=-a\times v+(n-a)\times\overline{v},
\end{equation}
together with the conservation of kinetic energy
\begin{equation}\label{spherecoke}
|v_{n}'|^{2}+|\overline{v}_{n}'|^{2}=|v|^{2}+|\overline{v}|^{2}
\end{equation}
for all $v, \overline{v}\in\mathbb{R}^{3}$. Moreover, in order that the hard spheres do not interpenetrate following collision, the scattering velocity variables must also satisfy the inequality
\begin{equation}\label{spherenonpen}
(v_{n}'-\overline{v}_{n}')\cdot n\geq 0 
\end{equation}
whenever the velocity variables $v, \overline{v}\in\mathbb{R}^{3}$ satisfy the collision parameter-dependent condition $(v-\overline{v})\cdot n\leq 0$. For given $v, \overline{v}\in\mathbb{R}^{3}$ and $n\in\mathbb{S}^{2}$, the semi-algebraic system (in the sense of real algebraic geometry: see Bochnak, Coste and Roy \cite{bochnak2013real}) comprising \eqref{spherecolm}--\eqref{spherenonpen} admits a unique solution $v_{n}', \overline{v}_{n}'\in\mathbb{R}^{3}$ given by \eqref{spherescat} above. We also remark in passing that the $n$-dependent smooth diffeomorphism of $\mathbb{R}^{6}$ effected by the transformation
\begin{equation}\label{unitjaccysphere}
\left(
\begin{array}{c}
v_{1} \\
v_{2} \\
v_{3} \\
\overline{v}_{1} \\
\overline{v}_{2} \\
\overline{v}_{3}
\end{array}
\right)\mapsto \left(
\begin{array}{c}
v_{n, 1}' \\
v_{n, 2}' \\
v_{n, 3}' \\
\overline{v}_{n, 1}' \\
\overline{v}_{n, 2}' \\
\overline{v}_{n, 3}'
\end{array}
\right)
\end{equation}
admits a Jacobian determinant of $-1$. This important consequence of the conservation laws yields, in particular, the well-known H-theorem for the Boltzmann equation: see (\cite{MR1307620}, Chapter 3.4), for instance.

The case of non-spherical particles is in stark contrast to that of non-rotational hard spheres, as there are infinitely-many `smooth' ways by which to resolve a collision between two convex non-spherical hard particles. Notably, the analogous $-1$-determinant property for the change of variables \eqref{unitjaccysphere} in the case of non-spherical particles is \emph{not} a consequence of the algebraic conservation laws of classical physics. Indeed, in Wilkinson \cite{wilkinson2020non}, it was shown that even if one assumes the assignment \eqref{unitjaccysphere} of `pre-' to `post-collisional' velocity variables to be a \emph{linear} map on Euclidean space, the conservation of linear momentum, angular momentum and kinetic energy are insufficient to guarantee the uniqueness of physical scattering for non-spherical hard particle dynamics. 
\subsubsection{Scattering for Non-spherical Particles}\label{scattynonsphere}
\begin{figure}
    \centering
    \includegraphics[scale=0.39]{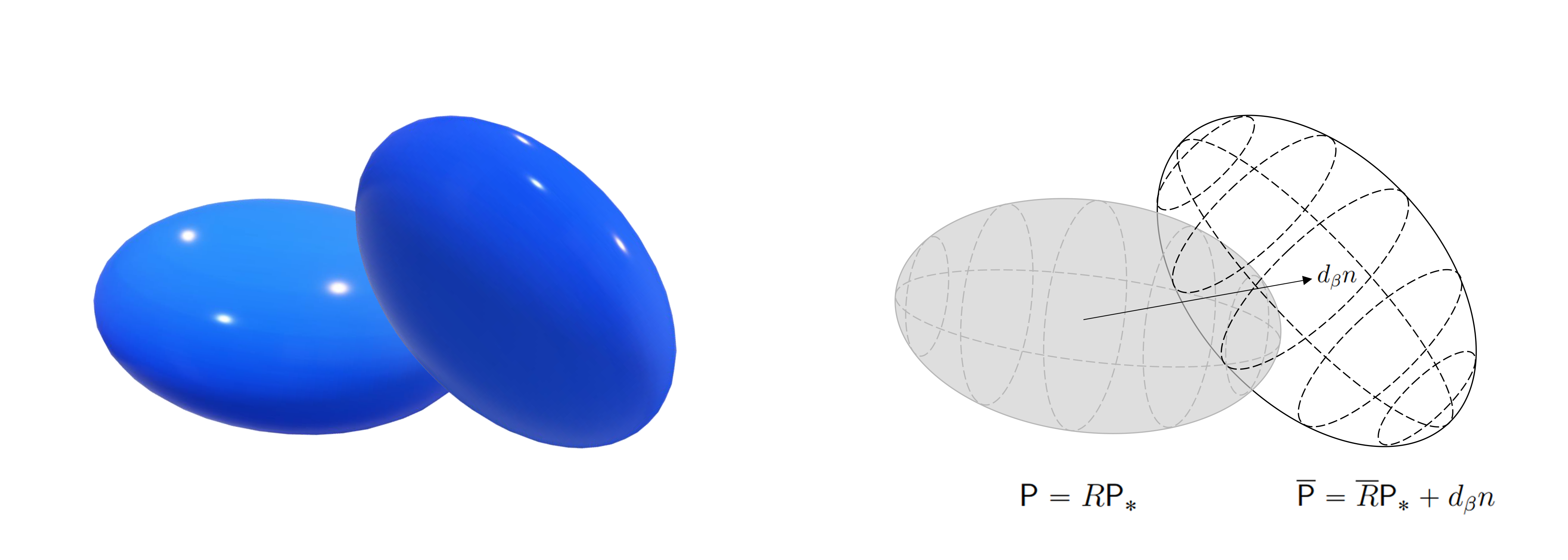}
    \caption{A collision configuration of two compact, convex subsets $\mathsf{P}$ and $\overline{\mathsf{P}}$ of $\mathbb{R}^{3}$, each of which is congruent to a given reference set $\mathsf{P}_{\ast}$. The matrices $R, \overline{R}\in\mathrm{SO}(3)$ represent the orientations of the two hard particles, $n\in\mathbb{S}^{2}$ represents the direction vector connecting the centre of mass of the unbarred particle to that of the barred, and $d_{\beta}>0$ denotes the distance of closest approach \eqref{distca3d}.}
    \label{fig:3dbluenonconcoll}
\end{figure}
It requires some effort to set up the framework within which collision invariants for non-spherical hard particle scattering can be defined. In order to articulate the above claims more precisely, we consider the collision between two hard particles $\mathsf{P}$ and $\overline{\mathsf{P}}$, each of which is congruent to a given compact, convex set $\mathsf{P}_{\ast}\subset\mathbb{R}^{3}$ termed the \emph{reference particle}, whose boundary surface $\partial\mathsf{P}_{\ast}$ is a differentiable manifold of class $\mathscr{C}^{1}$, i.e. $\mathsf{P}$ and $\overline{\mathsf{P}}$ are represented as
\begin{equation*}
\mathsf{P}=R\mathsf{P}_{\ast}  \quad \text{and}\quad \overline{\mathsf{P}}=\overline{R}\mathsf{P}_{\ast}+d_{\beta}n, 
\end{equation*}
where $\beta=(n, R, \overline{R})\in\mathbb{S}^{2}\times\mathrm{SO}(3)\times\mathrm{SO}(3)$ denotes the \emph{collision parameter} for non-spherical particles, and $d_{\beta}>0$ denotes the so-called \emph{distance of closest approach} given by
\begin{equation}\label{distca3d}
d_{\beta}:=\min\left\{
d>0\,:\,\mathrm{card}(R\mathsf{P}_{\ast}\cap (\overline{R}\mathsf{P}_{\ast}+dn))=0
\right\}.
\end{equation}
We refer the reader to Figure \ref{fig:3dbluenonconcoll} for an illustration of a collision configuration and its associated collision parameters. In what follows, we write $m>0$ and $J\in\mathbb{R}^{3\times 3}$ to denote the mass and the inertia matrix of the reference particle $\mathsf{P}_{\ast}$, respectively, namely
\begin{equation*}
m:=\int_{\mathsf{P}_{\ast}}\,dy\quad \text{and}\quad J:=\int_{\mathsf{P}_{\ast}}\left(|y|^{2}I-y\otimes y\right)\,dy.
\end{equation*}
Analogous to the case of hard spheres, we seek scattering linear velocities $v_{\beta}', \overline{v}_{\beta}'\in\mathbb{R}^{3}$ and angular velocities $\omega_{\beta}', \overline{\omega}_{\beta}'\in\mathbb{R}^{3}$ which obey the conservation of linear momentum
\begin{equation}\label{nonspherecolm}
mv_{\beta}'+m\overline{v}_{\beta}'=mv+m\overline{v},    
\end{equation}
the conservation of angular momentum (with respect to a point of measurement $a\in\mathbb{R}^{3}$) given by
\begin{equation}\label{nonspherecoam}
\begin{array}{c}
-ma\times v_{\beta}'+RJR^{T}\omega_{\beta}'+m(d_{\beta}n-a)\times\overline{v}_{\beta}'+\overline{R}J\overline{R}^{T}\overline{\omega}_{\beta}' \vspace{2mm} \\
=-ma\times v+RJR^{T}\omega+m(d_{\beta}n-a)\times\overline{v}+\overline{R}J\overline{R}^{T}\overline{\omega}, 
\end{array}
\end{equation}
as well as the conservation of kinetic energy
\begin{equation}\label{nonspherecoke}
\begin{array}{l}
m|v_{\beta}'|^{2}+RJR^{T}\omega_{\beta}'\cdot\omega_{\beta}'+m|\overline{v}_{\beta}'|^{2}+\overline{R}J\overline{R}^{T}\overline{\omega}_{\beta}'\cdot\overline{\omega}_{\beta}' \vspace{2mm} \\
=m|v|^{2}+RJR^{T}\omega\cdot\omega+m|\overline{v}|^{2}+\overline{R}J\overline{R}^{T}\overline{\omega}\cdot\overline{\omega} 
\end{array}
\end{equation}
for all given $v, \overline{v}\in\mathbb{R}^{3}$ and $\omega, \overline{\omega}\in\mathbb{R}^{3}$. Assuming that the underlying dynamics of $\mathsf{P}$ and $\overline{\mathsf{P}}$ is both left- and right-differentiable at all points of the real line, it can be shown that
\begin{equation}\label{leftandright}
\frac{d}{dt_{+}}\bigg|_{t=\tau}F(x(t), \overline{x}(t), R(t), \overline{R}(t))\geq 0 \quad \text{and}\quad \frac{d}{dt_{-}}\bigg|_{t=\tau}F(x(t), \overline{x}(t), R(t), \overline{R}(t))\leq 0,  
\end{equation}
where $F:\mathbb{R}^{3}\times\mathbb{R}^{3}\times\mathrm{SO}(3)\times\mathrm{SO}(3)\rightarrow\mathbb{R}$ is the auxiliary function defined by
\begin{equation*}
F(x, \overline{x}, R, \overline{R}):=|x-\overline{x}|^{2}-d_{(n(x, \overline{x}), R, \overline{R})}^{2},    
\end{equation*}
with $n(x, \overline{x})\in\mathbb{S}^{2}$ given by
\begin{equation*}
n(x, \overline{x}):=\frac{\overline{x}-x}{|\overline{x}-x|} 
\end{equation*}
for any $x, \overline{x}\in\mathbb{R}^{3}$ with $x\neq\overline{x}$, and $\tau\in\mathbb{R}$ is any time for which $F(x(\tau), \overline{x}(\tau), R(\tau), \overline{R}(\tau))=0$. The reader will note that the auxiliary function $F$ can be used to determine the admissible phase space for the dynamics of the two congruent hard particles $\mathsf{P}$ and $\overline{\mathsf{P}}$. For example, if we define
\begin{equation*}
\mathsf{P}(x, R):=R\mathsf{P}_{\ast}+x  
\end{equation*}
for $x\in\mathbb{R}^{3}$ and $R\in\mathrm{SO}(3)$, and the reference particle $\mathsf{P}_{\ast}\subset\mathbb{R}^{3}$ is the special case of a compact, \emph{strictly-convex} set, it holds that
\begin{equation*}
\mathrm{card}(\mathsf{P}(x, R)\cap\mathsf{P}(\overline{x}, \overline{R}))=0\quad \Longleftrightarrow\quad F(x, \overline{x}, R, \overline{R})>0,  
\end{equation*}
while
\begin{equation*}
\mathrm{card}(\mathsf{P}(x, R)\cap\mathsf{P}(\overline{x}, \overline{R}))=1\quad \Longleftrightarrow\quad F(x, \overline{x}, R, \overline{R})=0.
\end{equation*}
As a consequence of the above inequalities \eqref{leftandright}, it can be shown under the aforementioned regularity criterion on the underlying particle dynamics that the scattering velocity variables $v_{\beta}', \overline{v}_{\beta}', \omega_{\beta}', \overline{\omega}_{\beta}'\in\mathbb{R}^{3}$ necessarily satisfy the inequality
\begin{equation}\label{nonspherenonpen}
\begin{array}{c}
d_{\beta}n\cdot(\overline{v}_{\beta}'-v_{\beta}')-\nabla_{\mathbb{S}^{2}}d_{\beta}(I+n\otimes n)(\overline{v}_{\beta}'-v_{\beta}') \vspace{2mm}\\
-d_{\beta}D_{R}d_{\beta}R:\Omega_{\beta}'-d_{\beta}D_{\overline{R}}d_{\beta}\overline{R}:\overline{\Omega}_{\beta}'\geq 0    
\end{array}
\end{equation}
whenever $v, \overline{v}, \omega, \overline{\omega}\in\mathbb{R}^{3}$ satisfy the inequality 
\begin{equation}\label{nonsphereprecoll}
\begin{array}{c}
d_{\beta}n\cdot(\overline{v}-v)-\nabla_{\mathbb{S}^{2}}d_{\beta}(I+n\otimes n)(\overline{v}-v) \vspace{2mm} \\
-d_{\beta}D_{R}d_{\beta}R:\Omega-d_{\beta}D_{\overline{R}}d_{\beta}\overline{R}:\overline{\Omega}\leq 0,
\end{array}
\end{equation}
where $D_{R}d_{\beta}$ and $D_{\overline{R}}d_{\beta}$ denote the derivatives of $d_{\beta}$ with respect to its $R$- and $\overline{R}$-arguments, respectively. For given collision parameter $\beta$, we refer to any collection of velocities which obey \eqref{nonsphereprecoll} as \emph{pre-collisional} velocities, whilst any collection obeying \eqref{nonspherenonpen} as \emph{post-collisional} velocities. 

Unlike the case of hard spheres, for given $v, \overline{v}, \omega, \overline{\omega}\in\mathbb{R}^{3}$ and given $\beta\in\mathbb{S}^{2}\times\mathrm{SO}(3)\times\mathrm{SO}(3)$, the semi-algebraic system comprised of \eqref{nonspherecolm}--\eqref{nonspherenonpen} above does not admit a unique solution $v_{\beta}', \overline{v}_{\beta}', \omega_{\beta}', \overline{\omega}_{\beta}'\in\mathbb{R}^{3}$. In particular, the $-1$-Jacobian determinant property of the scattering map in the case of hard spheres, which is important for its role in the H-theorem, is not a consequence of the algebraic laws of classical mechanics \eqref{nonspherecolm}--\eqref{nonspherecoke} together with the semi-algebraic conditions \eqref{nonspherenonpen}--\eqref{nonsphereprecoll}. Therefore, motivated by the H-theorem in the case of hard spheres, if one assumes in addition that the $\beta$-dependent map on $\mathbb{R}^{12}$ defined pointwise by
\begin{equation}\label{bigmapnonsphere}
\left(
\begin{array}{c}
v_{1} \\
v_{2} \\
v_{3} \\
\overline{v}_{1} \\
\overline{v}_{2} \\
\overline{v}_{3} \\
\omega_{1} \\
\omega_{2} \\
\omega_{3} \\
\overline{\omega}_{1} \\
\overline{\omega}_{2} \\
\overline{\omega}_{3}
\end{array}
\right)
\mapsto\left(
\begin{array}{c}
v_{\beta, 1}' \\
v_{\beta, 2}' \\
v_{\beta, 3}' \\
\overline{v}_{\beta, 1}' \\
\overline{v}_{\beta, 2}' \\
\overline{v}_{\beta, 3}' \\
\omega_{\beta, 1}' \\
\omega_{\beta, 2}' \\
\omega_{\beta, 3}' \\
\overline{\omega}_{\beta, 1}' \\
\overline{\omega}_{\beta, 2}' \\
\overline{\omega}_{\beta, 3}'
\end{array}
\right)
\end{equation}
is a smooth diffeomorphism of $\mathbb{R}^{12}$ with Jacobian determinant $-1$, then it can be shown (see \cite{wilkinson2020non}) that there exist infinitely-many maps \eqref{bigmapnonsphere} which are both linear and which satisfy the semi-algebraic system \eqref{nonspherecolm}--\eqref{nonspherenonpen}. As such, in the case of non-spherical hard particles, the notion of collision invariant \emph{depends on the choice of scattering}, and is not a priori independent thereof. 

It will be helpful to establish some notation moving forward. In what follows, we shall employ the variable $V\in\mathbb{R}^{12}$ to denote the concatenated velocity variable 
\begin{equation*}
V=\left[
\begin{array}{c}
v \\
\overline{v} \\
\omega \\
\overline{\omega}
\end{array}
\right],    
\end{equation*}
where $v, \overline{v}, \omega, \overline{\omega}\in\mathbb{R}^{3}$, as well as $V_{\beta}'\in\mathbb{R}^{12}$ to denote the collision parameter-dependent scattering variable
\begin{equation*}
V_{\beta}'=\left[ 
\begin{array}{c}
v_{\beta}' \\
\overline{v}_{\beta}' \\
\omega_{\beta}' \\
\overline{\omega}_{\beta}'
\end{array}
\right].
\end{equation*}
We conclude this discussion with a definition which redresses the above objects in a framework more amenable to the techniques employed in this article.
\begin{defn}[Physical Scattering Map]
Suppose the collision parameter $\beta\in \mathbb{S}^{2}\times\mathrm{SO}(3)\times\mathrm{SO}(3)$ is given, and suppose further that the reference particle  $\mathsf{P}_{\ast}\subset\mathbb{R}^{3}$ is a compact, convex set with boundary manifold  $\partial\mathsf{P}_{\ast}\subset\mathbb{R}^{3}$ of class $\mathscr{C}^{1}$. We say that a map $\sigma_{\beta}:\mathbb{R}^{12}\rightarrow\mathbb{R}^{12}$ is a \textbf{physical scattering map} if and only if it is a classical solution of the Jacobian PDE given by
\begin{equation*}
\mathrm{det}\,D\sigma_{\beta}[V]=-1    
\end{equation*}
for all $V\in\mathbb{R}^{12}$, and in addition is subject to the algebraic constraints
\begin{equation*}
\sigma_{\beta}[V]\cdot\widehat{E}_{i}=V\cdot\widehat{E}_{i}
\end{equation*}
as well as
\begin{equation*}
\sigma_{\beta}[V]\cdot\alpha_{\beta}^{(i)}=V\cdot\alpha_{\beta}^{(i)}
\end{equation*}
for $i\in\{1, 2, 3\}$, and
\begin{equation*}
|M\sigma_{\beta}[V]|^{2}=|MV|^{2}    
\end{equation*}
for all $V\in\mathbb{R}^{12}$, where
\begin{equation*}
\widehat{E}_{1}:=\frac{1}{\sqrt{2}}\left(
\begin{array}{c}
1 \\
0 \\
0 \\
1 \\
0 \\
0 \\
0 \\
0 \\
0 \\
0 \\
0 \\
0 
\end{array}
\right), \quad \widehat{E}_{2}:=\frac{1}{\sqrt{2}}\left(
\begin{array}{c}
0 \\
1 \\
0 \\
0 \\
1 \\
0 \\
0 \\
0 \\
0 \\
0 \\
0 \\
0 
\end{array}
\right) \quad \text{and}\quad \widehat{E}_{3}:=\frac{1}{\sqrt{2}}\left(
\begin{array}{c}
0 \\
0 \\
1 \\
0 \\
0 \\
1 \\
0 \\
0 \\
0 \\
0 \\
0 \\
0 
\end{array}
\right),
\end{equation*}
together with
\begin{equation*}
\alpha_{\beta}^{(1)}:=\left(
\begin{array}{c}
0 \\
0 \\
0 \\
0 \\
md_{\beta}\varepsilon_{1j2}n_{j} \vspace{1mm} \\
md_{\beta}\varepsilon_{1j3}n_{j} \vspace{1mm} \\
R_{1j}J_{jk}R^{T}_{k1} \vspace{1mm} \\
R_{1j}J_{jk}R^{T}_{k2} \vspace{1mm} \\
R_{1j}J_{jk}R^{T}_{k3} \vspace{1mm} \\
\overline{R}_{1j}J_{jk}\overline{R}^{T}_{k1} \vspace{1mm} \\
\overline{R}_{1j}J_{jk}\overline{R}^{T}_{k2} \vspace{1mm} \\
\overline{R}_{1j}J_{jk}\overline{R}^{T}_{k3}
\end{array}
\right), \quad \alpha_{\beta}^{(2)}:=\left(
\begin{array}{c}
0 \\
0 \\
0 \\
md_{\beta}\varepsilon_{2j1}n_{j} \vspace{1mm} \\
0 \\
md_{\beta}\varepsilon_{2j3}n_{j} \vspace{1mm} \\
R_{2j}J_{jk}R^{T}_{k1} \vspace{1mm} \\
R_{2j}J_{jk}R^{T}_{k2} \vspace{1mm} \\
R_{2j}J_{jk}R^{T}_{k3} \vspace{1mm} \\
\overline{R}_{2j}J_{jk}\overline{R}^{T}_{k1} \vspace{1mm} \\
\overline{R}_{2j}J_{jk}\overline{R}^{T}_{k2} \vspace{1mm} \\
\overline{R}_{2j}J_{jk}\overline{R}^{T}_{k3}
\end{array}
\right) \quad \text{and}\quad \alpha_{\beta}^{(3)}:=\left(
\begin{array}{c}
0 \\
0 \\
0 \\
md_{\beta}\varepsilon_{3j1}n_{j} \vspace{1mm} \\
md_{\beta}\varepsilon_{3j2}n_{j} \vspace{1mm}\\
0 \\
R_{3j}J_{jk}R^{T}_{k1} \vspace{1mm} \\
R_{3j}J_{jk}R^{T}_{k2} \vspace{1mm} \\
R_{3j}J_{jk}R^{T}_{k3} \vspace{1mm} \\
\overline{R}_{3j}J_{jk}\overline{R}^{T}_{k1} \vspace{1mm} \\
\overline{R}_{3j}J_{jk}\overline{R}^{T}_{k2} \vspace{1mm} \\
\overline{R}_{3j}J_{jk}\overline{R}^{T}_{k3}
\end{array}
\right),        
\end{equation*}
where $\varepsilon_{ijk}$ denotes the Levi-Cevita symbol, and $M\in\mathbb{R}^{12\times 12}$ denotes the block mass-inertia matrix given by
\begin{equation*}
M:=\left(
\begin{array}{c | c | c | c}
\sqrt{m}I & 0 & 0 & 0 \\ \hline
0 & \sqrt{m}I & 0 & 0 \\ \hline
0 & 0 & \sqrt{J} & 0 \\ \hline
0 & 0 & 0 & \sqrt{J}
\end{array}
\right).
\end{equation*}
Moreover, the map $\sigma_{\beta}$ is also subject to the half-space mapping property that $\sigma_{\beta}[V]\cdot N_{\beta}\geq 0$ whenever $V\cdot N_{\beta}\leq 0$, where $N_{\beta}\in\mathbb{R}^{12}$ is the vector
\begin{equation*}
N_{\beta}:=\left[
\begin{array}{c}
-d_{\beta}n_{\beta}+\nabla_{\mathbb{S}^{2}}d_{\beta}-n_{\beta}\otimes n_{\beta}\nabla_{\mathbb{S}^{2}} d_{\beta} \vspace{2mm} \\
d_{\beta}n_{\beta}-\nabla_{\mathbb{S}^{2}}d_{\beta}+n_{\beta}\otimes n_{\beta}\nabla_{\mathbb{S}^{2}} d_{\beta}\vspace{2mm}  \\
d_{\beta}T(D_{R}d_{\beta}R^{T}) \vspace{2mm} \\
d_{\beta}T(D_{\overline{R}}\overline{R}^{T}) \\
\end{array}
\right],    
\end{equation*}
and $T:A\mapsto T(A)$ maps $\mathbb{R}^{3\times 3}$ to $\mathbb{R}^{3}$ pointwise as
\begin{equation*}
T(A):=\left(
\begin{array}{c}
A_{32}-A_{23} \vspace{1mm}\\
A_{13}-A_{31} \vspace{1mm} \\
A_{21}-A_{12}
\end{array}
\right)
\end{equation*}
for all $A\in\mathbb{R}^{3\times 3}$.
\end{defn}
\begin{rem}
As the above definition indicates, in all that follows, we suppress the dependence of a given scattering map $\sigma_{\beta}$ on the reference particle $\mathsf{P}_{\ast}$ that was used to construct it.
\end{rem}
With this definition in place, if $\sigma_{\beta}$ is a physical scattering map for all collision parameters $\beta\in\mathbb{S}^{2}\times\mathrm{SO}(3)\times\mathrm{SO}(3)$, we say that the collection
\begin{equation*}
\{\sigma_{\beta}\}_{\beta\in\mathbb{S}^{2}\times\mathrm{SO}(3)\times\mathrm{SO}(3)}
\end{equation*}
is a \textbf{physical scattering family}. In terms of this definition, under the assumption that the scattering map is linear, i.e. $\sigma_{\beta}$ is of the form
\begin{equation*}
\sigma_{\beta}[V]=S_{\beta}V
\end{equation*}
for some matrix $S_{\beta}\in\mathbb{R}^{12\times 12}$ and all $V\in\mathbb{R}^{12}$, the algebraic conservation laws take the form
\begin{equation*}
S_{\beta}^{T}\widehat{E}_{i}=\widehat{E}_{i}    
\end{equation*}
together with
\begin{equation*}
S_{\beta}^{T}\alpha_{\beta}^{(i)}=\alpha_{\beta}^{(i)}    
\end{equation*}
for $i\in\{1, 2, 3\}$, and 
\begin{equation*}
|MS_{\beta}M^{-1}V|^{2}=|V|^{2}
\end{equation*}
for all $V\in\mathbb{R}^{12}$. With the definitions of physical scattering maps and families in place, we can phrase the non-uniqueness result of \cite{wilkinson2020non} as the statement that if the reference particle $\mathsf{P}_{\ast}\subset\mathbb{R}^{3}$ is a compact, convex set with boundary manifold of class $\mathscr{C}^{1}$, for any fixed collision parameter $\beta$ there exist infinitely-many matrices $R_{\beta}\in\mathbb{R}^{12\times 12}$ such that the matrix $S_{\beta}:=M^{-1}R_{\beta}M$ effects a physical scattering map on $\mathbb{R}^{12}$. We are now in a position to define what we mean by a collision invariant in the case of non-spherical particle dynamics in this article.
\begin{rem}
At present, it is unknown if, for given collision parameter $\beta$, there exist any nonlinear physical scattering maps on Euclidean space. Assuming that solutions of the Jacobian equation are conservative vector fields, i.e. $\sigma_{\beta}=\nabla s_{\beta}$ for some $s_{\beta}\in C^{2}(\mathbb{R}^{12}, \mathbb{R})$, this may be redressed as the search for non-quadratic solutions of the associated Monge-Amp\`{e}re equation
\begin{equation*}
\mathrm{det}\,D^{2}s_{\beta}[V]=-1    
\end{equation*}
for all $V\in\mathbb{R}^{12}$, with additional constraints on $\nabla s_{\beta}$ arising from the conservation of linear momentum, angular momentum, and kinetic energy, as well as the inequalities arising from the non-penetration condition.
\end{rem}
\subsection{Collision Invariants for Non-spherical Hard Particles}
Owing to the non-uniqueness of physical scattering families for the dynamics of non-spherical particles, the notion of collision invariant in this case depends on the choice of physical scattering family. Indeed, in light of this observation, we shall adopt the following definition of collision invariant for the case of three-dimensional compact, convex reference particles $\mathsf{P}_{\ast}\subset\mathbb{R}^{3}$ in the sequel:
\begin{defn}[$\mathscr{S}$-Collision Invariants in Three Dimensions]
Let $\mathsf{P}_{\ast}\subset\mathbb{R}^{3}$ be a compact, convex set with boundary surface of class $\mathscr{C}^{1}$. Suppose a physical scattering family $\mathscr{S}=\{\sigma_{\beta}\}_{\beta\in\mathbb{S}^{2}\times\mathrm{SO}(3)\times\mathrm{SO}(3)}$ is given. We say that a measurable map $\varphi:\mathbb{R}^{3}\times\mathbb{R}^{3}\times\mathrm{SO}(3)\rightarrow\mathbb{R}$ is an \textbf{$\mathscr{S}$-collision invariant} if and only if it satisfies the functional equation
\begin{equation}\label{3dcollinv}
\varphi(v_{\beta}', \omega_{\beta}', R)+\varphi(\overline{v}_{\beta}', \overline{\omega}_{\beta}', \overline{R})=\varphi(v, \omega, R)+\varphi(\overline{v}, \overline{\omega}, \overline{R})      
\end{equation}
for $\mathscr{L}^{12}$-a.e. $[v, \overline{v}, \omega, \overline{\omega}]\in\mathbb{R}^{12}$ and $\mathscr{H}\otimes\mu\otimes\mu$-a.e. $(n, R, \overline{R})\in\mathbb{S}^{2}\times\mathrm{SO}(3)\times\mathrm{SO}(3)$, where 
\begin{equation*}
V_{\beta}':=\sigma_{\beta}[V]    
\end{equation*}
for $\mathscr{L}^{12}$-a.e. $V\in\mathbb{R}^{12}$.
\end{defn}
As the statement of our main theorem, namely Theorem \ref{maintheorem} below, will make clear, for a fixed choice of physical scattering family $\mathscr{S}=\{\sigma_{\beta}\}_{\beta\in\mathbb{S}^{2}\times\mathrm{SO}(3)\times\mathrm{SO}(3)}$, our aim is to find all possible $\mathscr{S}$-collision invariants of a given regularity. At first glance, characterisation of all maps satisfying \eqref{3dcollinv} is a problem of the theory of functional equations: see, for example, the monograph of Kuczma \cite{kuczma2009introduction}. However, it is possible to transform the identity \eqref{3dcollinv} to reveal the role of group theory in this problem. Indeed, if we define a new $\varphi$-dependent auxiliary map $\Phi:\mathbb{R}^{12}\times \mathrm{SO}(3)\times\mathrm{SO}(3)\rightarrow\mathbb{R}$ pointwise by
\begin{equation*}
\Phi(V, R, \overline{R}):=\varphi(v, \omega, R)+\varphi(\overline{v}, \overline{\omega}, \overline{R})    
\end{equation*}
for all $V=[v, \overline{v}, \omega, \overline{\omega}]\in\mathbb{R}^{12}$ and $(R, \overline{R})\in\mathrm{SO}(3)\times\mathrm{SO}(3)$, then \eqref{3dcollinv} becomes the statement that $\Phi$ satisfies
\begin{equation}\label{newandnice}
\Phi(\sigma_{\beta}[V], R, \overline{R})=\Phi(V, R, \overline{R})   
\end{equation}
for $V\in\mathbb{R}^{12}$ and $\beta\in\mathbb{S}^{2}\times\mathrm{SO}(3)\times\mathrm{SO}(3)$. By iteration of the action of the scattering map $\sigma_{\beta}$ on $V$, the above identity \eqref{newandnice} implies that $\Phi$ satisfies
\begin{equation*}
\Phi\left(\prod_{i=1}^{N}\sigma_{\beta(i)}[V], R, \overline{R}\right)=\Phi(V, R, \overline{R})
\end{equation*}
for any $N\geq 2$ and any collection of collision parameters $\{\beta(i)\}_{i=1}^{N}\subset\mathbb{S}^{2}\times\mathrm{SO}(3)\times\mathrm{SO}(3)$ with members of the form $\beta(i)=(n(i), R, \overline{R})$ for $\{n(i)\}_{i=1}^{N}\subset\mathbb{S}^{2}$. Intuitively, one `freezes' the orientation variables $R$ and $\overline{R}$ and in turn considers all possible collision configurations between two congruent particles with the frozen orientations. In particular, it holds that
\begin{equation*}
\Phi(gV, R, \overline{R})=\Phi(V, R, \overline{R})    
\end{equation*}
for any $g\in G(R, \overline{R})$, where $G(R, \overline{R})$ denotes the matrix group generated by the physical scattering family $\mathscr{S}$ with the orientation collision parameters $(R, \overline{R})$ fixed. Moving forward, we shall term $G(R, \overline{R})$ the $(R, \overline{R})$-dependent \textbf{scattering group} associated to the physical scattering family $\mathscr{S}$. In essence, the problem of characterisation of $\mathscr{S}$-collision invariants is a problem of characterising scalar invariants of the group action of the scattering groups $G(R, \overline{R})$ on $\mathbb{R}^{12}$.
\subsubsection{The Canonical Physical Scattering Family}
Formulated in this way, the problem of characterisation of collision invariants was tackled in \cite{saint2018collision} only in the case of what we shall call herein the \textbf{canonical physical scattering family} $\mathscr{S}_{\ast}=\{\sigma_{\beta}^{\ast}\}_{\beta}$ defined member-wise by
\begin{equation*}
\sigma_{\beta}^{\ast}[V]:=\underbrace{M^{-1}(I_{12}-2\widehat{N}_{\beta}\otimes\widehat{N}_{\beta})M}_{S_{\beta}^{\ast}:=}V    
\end{equation*}
for all $\beta\in\mathbb{S}^{2}\times\mathrm{SO}(3)\times\mathrm{SO}(3)$ and $V\in\mathbb{R}^{12}$. We employ the term \emph{canonical} as it is the scattering which is employed without comment in many major studies of physical hard-particle dynamics. Indeed, we invite the reader to consult the well-known work of Donev, Torquato and Stillinger (\cite{donev2005neighbor}, Section 3.2) as an important example. Moreover, it may be considered canonical as it `reduces essentially' to the scattering map for non-rotational hard spheres when the reference particle $\mathsf{P}_{\ast}$ is chosen to be a sphere of any radius. Concretely, when the reference particle is chosen to be $\mathsf{P}_{\ast}=B(0, r)$ for any $r>0$, it can be shown that the generating matrix $S_{\beta}^{\ast}$ reduces to the matrix given in block form by
\begin{equation*}
S_{\beta}^{\ast}=M^{-1}\left(
\begin{array}{c | c}
I_{6}-2\widehat{\nu}_{n}\otimes\widehat{\nu}_{n} & 0_{6}\\ \hline
0_{6} & I_{6}
\end{array}
\right)M
\end{equation*}
for all $\beta=(n, R, \overline{R})\in\mathbb{S}^{2}\times\mathrm{SO}(3)\times\mathrm{SO}(3)$, where $\widehat{\nu}_{n}\in\mathbb{S}^{5}$ is the vector defined in \eqref{bigennspherenormal} above. In particular, using the theory of Ballard \cite{ballard2000dynamics}, the dynamics on the tangent bundle of the manifold with boundary 
\begin{equation*}
\left\{
X=[x, \overline{x}, R, \overline{R}]\in\mathbb{R}^{3}\times\mathbb{R}^{3}\times\mathrm{SO}(3)\times\mathrm{SO}(3)\,:\,\mathscr{L}^{3}((R\mathsf{P}_{\ast}+x)\cap(\overline{R}\mathsf{P}_{\ast}+\overline{x}))=0\right\}
\end{equation*}
constructed using the physical scattering family $\mathscr{S}_{\ast}$ is precisely canonical physical hard sphere dynamics which `acts as the identity on the angular velocity variables'.
\subsubsection{Reflections and the Work of Eaton and Perlman, and Viterbo}
The reader will note that the canonical physical scattering map $\sigma_{\beta}^{\ast}\in\mathrm{O}(12)$ above is generated by a matrix which is conjugate by the mass-inertia matrix $M$ to a reflection matrix. Owing to this reflection structure, the group $G(R, \overline{R})$ can be readily characterised using the work of Eaton and Perlman \cite{MR0463329}, or that of Viterbo in the appendix of \cite{saint2018collision}. Let us state the version of the result we shall employ that is taken from the appendix in \cite{saint2018collision}.
\begin{thm}[\cite{saint2018collision}, Appendix]
Suppose $q\geq 2$ is an integer, and that $\mathcal{X}$ is a connected topological space which is not reduced to a point. Let $\varrho:\mathcal{X}\rightarrow\mathbb{S}^{q-1}$ be a continuous map with the property that
\begin{equation*}
\mathrm{span}\{\varrho(x)\,:\,x\in\mathcal{X}\}=\mathbb{R}^{q}.    
\end{equation*}
It follows that the group generated by the associated set of reflections in $\mathrm{O}(q)$ given by
\begin{equation*}
\{I_{q}-2\varrho(x)\otimes\varrho(x)\,:\,x\in\mathcal{X}\}\subset\mathrm{O}(q)    
\end{equation*}
is the entire orthogonal group $O(q)$.
\end{thm}
Indeed, under technical assumptions on the reference particle $\mathsf{P}_{\ast}$, using this result it can be shown that for any $(R, \overline{R})\in\mathrm{SO}(3)\times\mathrm{SO}(3)$, the scattering group $G(R, \overline{R})\subset\mathrm{O}(12)$ is homomorphic to the orthogonal group $\mathrm{O}(9)$. As a consequence of characterising $G(R, \overline{R})$, it can then be shown that any $\mathscr{S}_{\ast}$-collision invariant $\varphi$ is necessarily of the form
\begin{equation}\label{3dcollinvstat}
\varphi(v, \omega, R)=a(R)+b\cdot v+c\left(m|v|^{2}+R J R^{T}\omega\cdot\omega\right)    
\end{equation}
for all $v, \omega\in\mathbb{R}^{3}$ and all $R\in\mathrm{SO}(3)$, where $a:\mathrm{SO}(3)\rightarrow\mathbb{R}$ is a measurable map, $b\in\mathbb{R}^{3}$ is a constant vector, and $c\in\mathbb{R}$ is a constant. Indeed, \eqref{3dcollinvstat} is a quantification of the claim \eqref{wordymaintheorem} made in the introduction. However, it is \emph{not} the case that all physical scattering maps induced by matrices are conjugate to a reflection matrix. In those cases, one does not have the technology of Eaton and Perlman and Viterbo to characterise group $G(R, \overline{R})$ in general. The primary contribution of this paper is that we develop a method which does not require physical scattering maps to be generated by a matrix that is conjugate to a reflection matrix. 
\subsection{The Case of Non-spherical Hard Particles in 2D}
\begin{figure}
    \centering
    \includegraphics[scale=0.45]{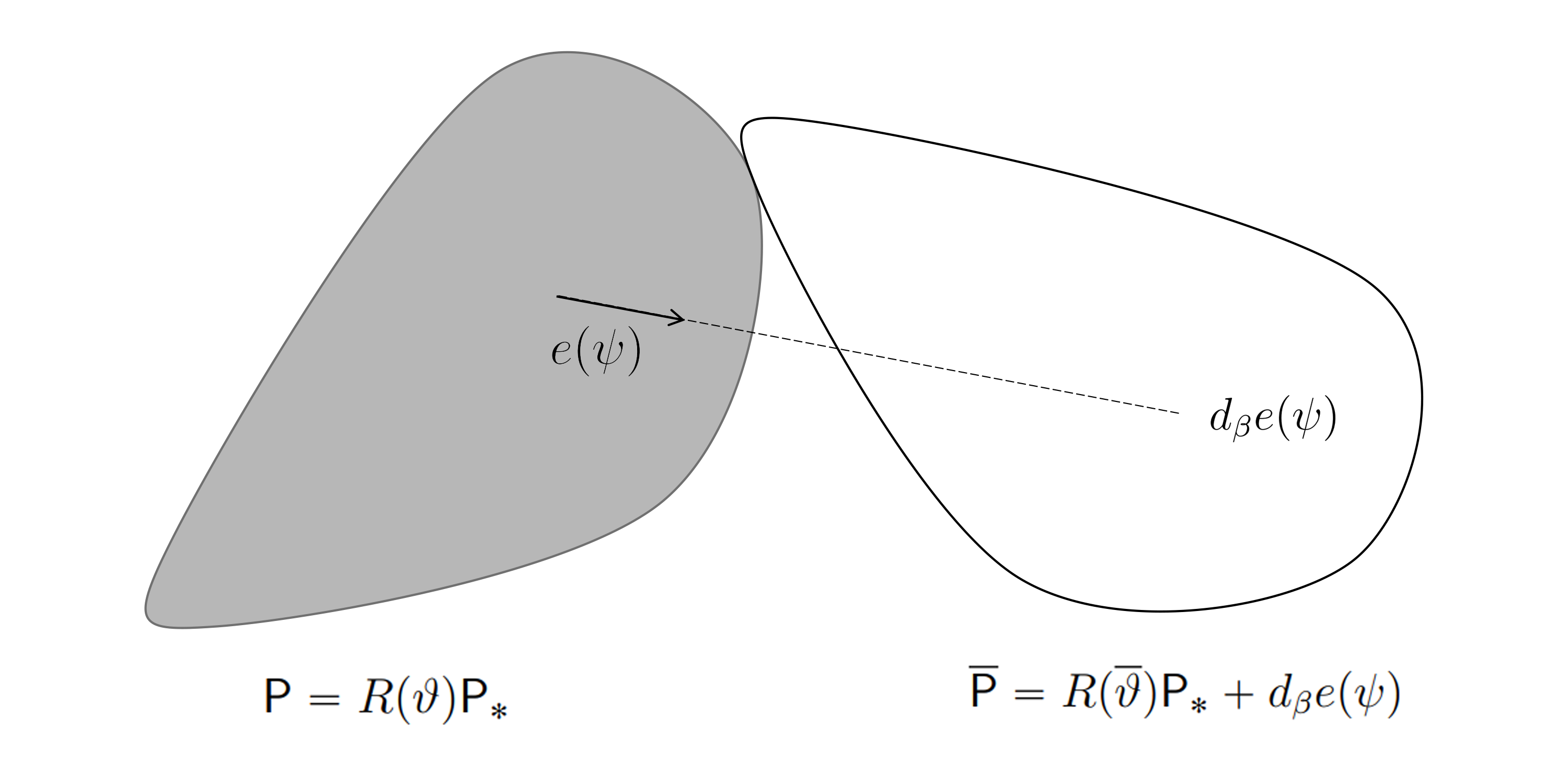}
    \caption{A collision configuration of two compact, convex subsets $\mathsf{P}$ and $\overline{\mathsf{P}}$ of $\mathbb{R}^{2}$, each of which is congruent to a given reference set $\mathsf{P}_{\ast}$. The elevation angle $\psi\in\mathbb{S}^{1}$ determines the direction vector $e(\psi):=(\cos\psi, \sin\psi)$ directed from the centre of the unbarred particle to that of the barred, $\vartheta,\overline{\vartheta}\in\mathbb{S}^{1}$ denote the orientations of the particles, whilst $d_{\beta}>0$ denotes the distance of closest approach \eqref{distca2d}.}
    \label{fig:2dcoll}
\end{figure}
Owing to the fact that the dimensionality of the problem for three-dimensional compact, convex particles serves to obscure the essential algebro-geometric argument we employ in this work, as stated in the introduction, we have elected to state and prove our main theorem (Theorem \ref{maintheorem}) for collision invariants in the notationally-simpler case of two-dimensional compact, convex hard reference particles $\mathsf{P}_{\ast}\subset\mathbb{R}^{2}$. Nevertheless, our methods extend without essential difficulty to the three-dimensional case.
\subsubsection{Set-up and Notation}
In what follows, mirroring section \ref{scattynonsphere} above in the three-dimensional case, we shall consider the collision between two hard particles $\mathsf{P}, \overline{\mathsf{P}}\subset\mathbb{R}^{2}$, both of which are congruent to a compact, convex set $\mathsf{P}_{\ast}\subset\mathbb{R}^{2}$ that we once again term the \emph{reference particle}, whose boundary $\partial\mathsf{P}_{\ast}$ admits the structure of a differentiable curve of class $\mathscr{C}^{1}$. As such, in the case that $\mathsf{P}\cap\overline{\mathsf{P}}\neq \varnothing$ yet $\mathscr{L}^{2}(\mathsf{P}\cap\overline{\mathsf{P}})=0$, it holds that
\begin{equation*}
\mathsf{P}=R(\vartheta)\mathsf{P}_{\ast}\quad \text{and}\quad \overline{\mathsf{P}}=R(\overline{\vartheta})\mathsf{P}_{\ast}+d_{\beta}e(\psi)    
\end{equation*}
for some $\beta=(\psi, \vartheta, \overline{\vartheta})\in\mathbb{T}^{3}$, where $d_{\beta}>0$ denotes the \emph{distance of closest approach} given by
\begin{equation}\label{distca2d}
d_{\beta}:=\min\left\{
d>0\,:\,\mathrm{card}(R(\vartheta)\mathsf{P}_{\ast}\cap(R(\overline{\vartheta})\mathsf{P}_{\ast}+de(\psi)))=0,
\right\},
\end{equation}
the rotation matrix $R(\theta)\in\mathrm{SO}(2)$ is defined for given $\theta\in\mathbb{S}^{1}$ by
\begin{equation*}
R(\theta):=\left(
\begin{array}{cc}
\cos\theta & -\sin\theta \\
\sin\theta & \cos\theta
\end{array}
\right),
\end{equation*}
and where $e(\psi)\in\mathbb{S}^{1}$ denotes the unit vector 
\begin{equation*}
e(\psi)=\left(
\begin{array}{c}
\cos\psi \\
\sin\psi
\end{array}
\right).
\end{equation*}
The collision parameter $\beta\in\mathbb{T}^{3}$ that characterises the collision between two congruent compact, convex particles is comprised of the \emph{elevation angle} $\psi\in\mathbb{S}^{1}$, as well as the \emph{particle orientations} $\vartheta, \overline{\vartheta}\in\mathbb{S}^{1}$ of $\mathsf{P}, \overline{\mathsf{P}}$ with respect to $\mathsf{P}_{\ast}$, respectively. The reader may consult Figure \ref{fig:2dcoll} for an illustration of a collision configuration and its corresponding collision parameters. We write $m>0$ and $J>0$ to denote the mass and scalar inertia of the reference particle $\mathsf{P}_{\ast}\subset\mathbb{R}^{2}$, namely
\begin{equation*}
m:=\int_{\mathsf{P}_{\ast}}\,dy\quad\text{and}\quad J:=\int_{\mathsf{P}_{\ast}}|y|^{2}\,dy.    
\end{equation*}
We now provide a definition of what we mean by a \emph{physical scattering map} in the case of the collision of two-dimensional compact, convex particles. In all that follows, following the convention that was adopted in three dimensions, we use the variable $V\in\mathbb{R}^{6}$ to denote the concatenated velocity variable given by
\begin{equation*}
V=\left[
\begin{array}{c}
v \\
\overline{v} \\
\omega \\
\overline{\omega}
\end{array}
\right],
\end{equation*}
for $v, \overline{v}\in\mathbb{R}^{2}$ and $\omega, \overline{\omega}\in\mathbb{R}$, as well as the variable $V_{\beta}'\in\mathbb{R}^{6}$ to denote the concatenated scattering velocity variable as follows:
\begin{equation*}
V_{\beta}'=\left[
\begin{array}{c}
v_{\beta}' \\
\overline{v}_{\beta}' \\
\omega_{\beta}' \\
\overline{\omega}_{\beta}'
\end{array}
\right].
\end{equation*}
We now define what we mean by a physical scattering map for the resolution of a collision between two compact, convex hard particles in two dimensions.
\begin{defn}[Physical Scattering Map in Two Dimensions]\label{physscatdeftwo}
Suppose a collision parameter $\beta\in\mathbb{T}^{3}$ is given, and suppose further that $\mathsf{P}_{\ast}\subset\mathbb{R}^{2}$ is a compact, convex set with boundary curve $\partial\mathsf{P}_{\ast}\subset\mathbb{R}^{2}$ of class $\mathscr{C}^{1}$. We say that a map $\sigma_{\beta}:\mathbb{R}^{6}\rightarrow\mathbb{R}^{6}$ is a \textbf{physical scattering map} if and only if it is a classical solution of the Jacobian equation given by
\begin{equation*}
\mathrm{det}\,D\sigma_{\beta}[V]=-1    
\end{equation*}
for all $V\in\mathbb{R}^{6}$, and in addition is subject to the algebraic constraints
\begin{equation*}
\sigma_{\beta}[V]\cdot\widehat{E}_{i}=V\cdot \widehat{E}_{i}    
\end{equation*}
for $i\in\{1, 2\}$, as well as
\begin{equation*}
\sigma_{\beta}[V]\cdot\alpha_{\beta}=V\cdot\alpha_{\beta}
\end{equation*}
and
\begin{equation*}
|M\sigma_{\beta}[V]|^{2}=|MV|^{2}    
\end{equation*}
for all $V\in\mathbb{R}^{6}$, where
\begin{equation*}
\widehat{E}_{1}:=\frac{1}{\sqrt{2}}\left(
\begin{array}{c}
1 \\
0 \\
1 \\
0 \\
0 \\
0
\end{array}
\right) \quad \text{and}\quad \widehat{E}_{2}:=\frac{1}{\sqrt{2}}\left(
\begin{array}{c}
0 \\
1 \\
0 \\
1 \\
0 \\
0
\end{array}
\right), 
\end{equation*}
and
\begin{equation*}
\alpha_{\beta}:=\left[
\begin{array}{c}
0 \\
-d_{\beta}e(\psi)^{\perp} \\
J \\
J \\
\end{array}
\right],
\end{equation*}
with $M\in\mathbb{R}^{6\times 6}$ the mass-inertia matrix given by
\begin{equation*}
M:=\mathrm{diag}(\sqrt{m}, \sqrt{m}, \sqrt{m}, \sqrt{m}, \sqrt{J}, \sqrt{J}).
\end{equation*}
Moreover, the map $\sigma_{\beta}$ is also subject to the half-space mapping property that $\sigma_{\beta}[V]\cdot \widehat{N}_{\beta}\geq 0$ whenever $V\cdot \widehat{N}_{\beta}\leq 0$, where $\widehat{N}_{\beta}\in\mathbb{S}^{5}$ is the (up to sign) unit normal to the 0-level-set hypersurface
\begin{equation*}
\left\{
Y\in\mathbb{R}^{6}\,:\,F(Y)=0
\right\}
\end{equation*}
at the point $[0, d_{(\psi(Y), y_{5}, y_{6})}e(\psi(Y)), y_{5}, y_{6}]\in\mathbb{R}^{6}$, where $F:\mathbb{R}^{6}\rightarrow\mathbb{R}$ is defined pointwise by
\begin{equation*}
F(Y):=(y_{1}-y_{3})^{2}+(y_{2}-y_{4})^{2}-d^{2}_{(\psi(Y), y_{5}, y_{6})}    
\end{equation*}
for all $Y=(y_{1}, y_{2}, y_{3}, y_{4}, y_{5}, y_{6})\in\mathbb{R}^{6}$, where $\psi(Y)$ is given by
\begin{equation*}
\psi(Y):=\mathrm{arctan}\left(\frac{y_{4}-y_{2}}{y_{3}-y_{1}}\right).
\end{equation*}
\end{defn}
As is the case in three dimensions, for a given collision parameter $\beta\in\mathbb{T}^{3}$, physical scattering maps are not unique. However, the difference in dimension, as we shall discover in the sequel, means that there can only be finitely-many solutions of the Jacobian equation for physical scattering in two dimensions. With this definition in place, we can now state what we mean by a collision invariant in the context of two-dimensional scattering.
\begin{defn}[$\mathscr{S}$-Collision Invariants in Two Dimensions] Let $\mathsf{P}_{\ast}\subset\mathbb{R}^{2}$ be a compact, convex set with boundary curve $\partial\mathsf{P}_{\ast}$ of class $\mathscr{C}^{1}$. Suppose a physical scattering family $\mathscr{S}$ is given. We say that a measurable map $\varphi:\mathbb{R}^{2}\times\mathbb{R}\times\mathbb{S}^{1}\rightarrow\mathbb{R}$ is an \textbf{$\mathscr{S}$-collision invariant} if and only if it satisfies the functional equation
\begin{equation*}
\varphi(v_{\beta}', \omega_{\beta}', \vartheta)+\varphi(\overline{v}_{\beta}', \overline{\omega}_{\beta}', \overline{\vartheta})=\varphi(v, \omega, \vartheta)+\varphi(\overline{v}, \overline{\omega}, \overline{\vartheta})
\end{equation*}
for $\mathscr{L}^{6}$-a.e. $[v, \overline{v}, \omega, \overline{\omega}]\in\mathbb{R}^{6}$ and $\nu$-a.e. $\beta\in\mathbb{T}^{3}$, where $\nu$ denotes the Haar measure on $\mathbb{T}^{3}$.
\end{defn}
\subsubsection{Full Enumeration of Physical Scattering Maps in 2D}
It was proved in \cite{wilkinson2020non} that, in the case of two-dimensional compact, convex hard particles $\mathsf{P}_{\ast}\subset\mathbb{R}^{2}$, for each fixed collision parameter $\beta\in\mathbb{T}^{3}$ there exist only \emph{two} physical scattering maps which are linear on $\mathbb{R}^{6}$. These maps are given by
\begin{equation*}
\sigma_{\beta}^{\ast}[V]:=M^{-1}(I-2\widehat{N}_{\beta}\otimes\widehat{N}_{\beta})MV    
\end{equation*}
to which we refer as the \textbf{canonical} physical scattering map, and also
\begin{equation*}
\sigma_{\beta}^{\times}[V]:=M^{-1}(2\widehat{E}_{1}\otimes\widehat{E}_{1}+2\widehat{E}_{2}\otimes \widehat{E}_{2}+2\widehat{E}_{\beta}\otimes\widehat{E}_{\beta}-I_{6})MV,    
\end{equation*}
to which we refer as the \textbf{non-canonical} physical scattering map, where $\widehat{E}_{1}, \widehat{E}_{2}\in\mathbb{S}^{5}$ are the unit vectors
\begin{equation*}
\widehat{E}_{1}:=\frac{1}{\sqrt{2}}\left(
\begin{array}{c}
1 \\
0 \\
1 \\
0 \\
0 \\
0
\end{array}
\right) \quad \text{and} \quad \widehat{E}_{2}:=\frac{1}{\sqrt{2}}\left(
\begin{array}{c}
0 \\
1 \\
0 \\
1 \\
0 \\
0
\end{array}
\right),
\end{equation*}
while $\widehat{E}_{\beta}\in\mathbb{S}^{5}$ is the collision parameter-dependent unit vector given by
\begin{equation*}
\widehat{E}_{\beta}:=\frac{1}{\sqrt{2md_{\beta}^{2}+8J}}\left(
\begin{array}{c}
\sqrt{m}d_{\beta}\sin\psi \\
-\sqrt{m}d_{\beta}\cos\psi \\
-\sqrt{m}d_{\beta}\sin\psi \\
\sqrt{m}d_{\beta}\cos\psi \\
2\sqrt{J} \\
2\sqrt{J}
\end{array}
\right).
\end{equation*}
In turn, we denote the associated \textbf{canonical physical scattering family} by $\mathscr{S}_{\ast}$, namely
\begin{equation*}
\mathscr{S}_{\ast}:=\{\sigma_{\beta}^{\ast}\}_{\beta\in\mathbb{T}^{3}},    
\end{equation*}
while we denote the associated \textbf{non-canonical physical scattering family} by $\mathscr{S}_{\times}$, i.e.
\begin{equation*}
\mathscr{S}_{\times}:=\{\sigma_{\beta}^{\times}\}_{\beta\in\mathbb{T}^{3}}.    
\end{equation*}
In general, the method one employs to characterise $\mathscr{S}$-collision invariants depends on the structure of the scattering maps in $\mathscr{S}$. In the case of non-canonical scattering, it is evident that $\sigma_{\beta}^{\times}$ as a matrix is not conjugate to a reflection matrix. As such, in order to employ the results of Viterbo in (\cite{saint2018collision}, Appendix), we search for methods appropriate to the structure of the members of the non-canonical physical scattering family $\mathscr{S}_{\times}$ itself. 
\subsubsection{Physical Interpretation of Non-uniqueness}
When written in their component forms, as opposed to their matrix forms, the physical scattering families readily admit a physical distinction from one another. Indeed, in the case of canonical scattering, it can be shown that for any $V\in\mathbb{R}^{6}$ the scattering variable $V_{\beta}'=\sigma_{\beta}^{\ast}[V]$ is written element-wise as
\begin{equation*}
\left\{
\begin{array}{l}
\displaystyle v_{\beta}':=v-\frac{1}{m\Lambda_{\beta}}(v+\omega p_{\beta}^{\perp}-\overline{v}-\overline{\omega}q_{\beta}^{\perp})\cdot n_{\beta} n_{\beta},\vspace{2mm}\\
\displaystyle  \overline{v}_{\beta}':=\overline{v}+\frac{1}{m\Lambda_{\beta}}(v+\omega p_{\beta}^{\perp}-\overline{v}-\overline{\omega}q_{\beta}^{\perp})\cdot n_{\beta} n_{\beta}, \vspace{2mm}\\
\displaystyle \omega_{\beta}':=\omega-\frac{1}{J\Lambda_{\beta}}(v+\omega p_{\beta}^{\perp}-\overline{v}-\overline{\omega}q_{\beta}^{\perp})\cdot n_{\beta}p_{\beta}^{\perp}\cdot n_{\beta}, \vspace{2mm}\\
\displaystyle \overline{\omega}_{\beta}':=\overline{\omega}+\frac{1}{J\Lambda_{\beta}}(v+\omega p_{\beta}^{\perp}-\overline{v}-\overline{\omega}q_{\beta}^{\perp})\cdot n_{\beta}q_{\beta}^{\perp}\cdot n_{\beta}
\end{array}
\right.
\end{equation*}
where $\Lambda_{\beta}>0$ is defined to be
\begin{equation*}
\Lambda_{\beta}:=\left|
\frac{2}{m}+\frac{1}{J}(p_{\beta}^{\perp}\cdot n_{\beta})^{2}+\frac{1}{J}(q_{\beta}^{\perp}\cdot n_{\beta})^{2}
\right|^{2},
\end{equation*}
and $p_{\beta}, q_{\beta}, n_{\beta}\in\mathbb{R}^{2}$ are the so-called \emph{collision vector}, \emph{conjugate collision vector}, and \emph{collision normal}, respectively, as defined in (\cite{saint2018collision}, Section 3.3). However, in the case of non-canonical scattering, $V_{\beta}'=\sigma_{\beta}^{\times}[V]$ is element-wise given by
\begin{equation*}
\left\{
\begin{array}{l}
\displaystyle v_{\beta}':=\overline{v}+\frac{1}{md_{\beta}^{2}+4J}\left(md_{\beta}v-2J\omega e(\psi)^{\perp}-m d_{\beta}\overline{v}-2J\overline{\omega}e(\psi)^{\perp}\right)\cdot e(\psi)^{\perp}d_{\beta}e(\psi)^{\perp}, \vspace{2mm}\\
\displaystyle  \overline{v}_{\beta}':=v-\frac{1}{md_{\beta}^{2}+4J}\left(md_{\beta}v-2J\omega e(\psi)^{\perp}-m d_{\beta}\overline{v}-2J\overline{\omega}e(\psi)^{\perp}\right)\cdot e(\psi)^{\perp}d_{\beta}e(\psi)^{\perp}, \vspace{2mm}\\
\displaystyle \omega_{\beta}':=-\omega-\frac{2}{md_{\beta}^{2}+4J}\left(md_{\beta}v-2J\omega e(\psi)^{\perp}-m d_{\beta}\overline{v}-2J\overline{\omega}e(\psi)^{\perp}\right)\cdot e(\psi)^{\perp}, \vspace{2mm}\\
\displaystyle \overline{\omega}_{\beta}':=-\overline{\omega}-\frac{2}{md_{\beta}^{2}+4J}\left(md_{\beta}v-2J\omega e(\psi)^{\perp}-m d_{\beta}\overline{v}-2J\overline{\omega}e(\psi)^{\perp}\right)\cdot e(\psi)^{\perp}.
\end{array}
\right.
\end{equation*}
As such, in the case of canonical scattering, the direction of linear impulse at collision is proportional to $n_{\beta}\in\mathbb{S}^{1}$, the normal at the point of contact at collision. However, in the case of non-canonical scattering, the direction of linear impulse is proportional to $e(\psi)^{\perp}\in\mathbb{S}^{1}$, the direction orthogonal to that connecting the centres of mass of the two hard particles at collision.
\begin{rem}
In \cite{wilkinson2020non}, the case of scattering maps which obey the Jacobian equation
\begin{equation*}
\mathrm{det}\,D\sigma_{\beta}[V]=1    
\end{equation*}
for all $V\in\mathbb{R}^{6}$ was also treated. Indeed, it was shown that even in the case of two-dimensional hard particles, there exist infinitely-many solutions of this equation which also obey the algebraic conservation laws of classical physics. We do not consider collision invariants for their physical scattering families in this article.
\end{rem}
\subsection{Rigorous Statement of Main Results}
In the sections leading up to the statement of our main result below, we have placed an emphasis on the fact that group theory plays a significant role in the characterisation of collision invariants, no matter the physical scattering family one studies. The general problem of characterisation of $\mathscr{S}$-collision invariants reduces, essentially, to the characterisation of the scattering (Lie) groups $G(\vartheta, \overline{\vartheta})$ given by
\begin{equation*}
G(\vartheta, \overline{\vartheta}):=\langle\{\sigma_{\beta}\}_{\psi\in\mathbb{S}^{1}} \rangle
\end{equation*}
for each pair of orientations $(\vartheta, \overline{\vartheta})\in\mathbb{T}^{2}$. We claim that, within the context in which we work, it is in general \emph{easier} to characterise the Lie algebra $\mathfrak{g}(\vartheta, \overline{\vartheta})$ of scattering Lie groups $G(\vartheta, \overline{\vartheta})$ than it is to characterise the scattering Lie groups themselves. We do this by `probing' the Lie algebra of $G(\vartheta, \overline{\vartheta})$ through considering only paths on $G(\vartheta, \overline{\vartheta})$ determined by $\{\sigma_{\beta}\}_{\psi\in\mathbb{S}^{1}}$, namely curves of matrices determined by the scattering family $\mathscr{S}$ by fixing the orientation pair $(\vartheta, \overline{\vartheta})$. Details of this `probing' in the more familiar case of non-rotational hard spheres are outlined in section \ref{hardsphereproof} below.

We now state the main result of this article, which is the mathematical formulation of statement \eqref{wordymaintheorem} above.
\begin{thm}[Characterisation of Collision Invariants for Non-canonical Physical Scattering $\mathscr{S}_{\times}$ in Two Dimensions]\label{maintheorem}
Suppose that $\mathsf{P}_{\ast}\subset\mathbb{R}^{2}$ is a compact, convex set whose boundary $\partial\mathsf{P}_{\ast}$ is of class $\mathscr{C}^{1}$. Suppose that $\varphi\in C^{1}(\mathbb{R}^{2}\times\mathbb{R}\times\mathbb{S}^{1})$ is an $\mathscr{S}_{\times}$-collision invariant, i.e. it satisfies the functional identity
\begin{equation*}
\varphi(v_{\beta}', \omega_{\beta}', \vartheta)+\varphi(\overline{v}_{\beta}', \overline{\omega}_{\beta}', \overline{\vartheta})=\varphi(v, \omega, \vartheta)+\varphi(\overline{v}, \overline{\omega}, \overline{\vartheta})
\end{equation*}
for all $v, \overline{v}\in\mathbb{R}^{2}$, $\omega, \overline{\omega}\in\mathbb{R}$, and all $\beta=(\psi, \vartheta, \overline{\vartheta})\in\mathbb{T}^{3}$, where the non-canonical scattering velocity variables $v_{\beta}', \overline{v}_{\beta}'\in\mathbb{R}^{2}$ and $\omega_{\beta}', \overline{\omega}_{\beta}'\in\mathbb{R}$ are given by $V_{\beta}'=\sigma_{\beta}^{\times}[V]$. It follows that $\varphi$ is necessarily of the form
\begin{equation*}
\varphi(v, \omega, \vartheta)=a(\vartheta)+b\cdot v+c\left(m|v|^{2}+J\omega^{2}\right)    
\end{equation*}
for some $a\in C^{1}(\mathbb{S}^{1})$, $b\in\mathbb{R}^{2}$, and $c\in\mathbb{R}$.
\end{thm}
We remark that, for our main result, we can relax the hypotheses placed in Saint-Raymond and Wilkinson \cite{saint2018collision} on $\mathsf{P}_{\ast}\subset\mathbb{R}^{2}$ in the case $\mathscr{S}_{\ast}$-collision invariants. Indeed, our result for $\mathscr{S}_{\times}$-collision invariants holds for compact, convex reference sets $\mathsf{P}_{\ast}$, rather than simply \emph{strictly convex} reference sets with symmetry. 
\subsection{Further Notation}
In what follows, for any integer $N\geq 1$, $\mathscr{L}^{N}$ denotes the Lebesgue measure on $\mathbb{R}^{N}$, $\mathscr{H}$ denotes the Hausdorff measure on $\mathbb{S}^{2}$, $\mu$ denotes the Haar measure on $SO(3)$, and $\nu$ denotes the Haar measure on $\mathbb{T}^{3}$. 

For any integer $q\geq 2$, $\mathrm{SO}(q)$ and $\mathrm{O}(q)$ denote the special orthogonal group and orthogonal group which act on $\mathbb{R}^{q}$, respectively. Moreover, $\mathfrak{so}(q)$ denotes the Lie algebra of skew-symmetric $q\times q$ matrices with real entries. 

For any integer $N\geq 2$, $I_{N}$ denotes the identity matrix in $\mathbb{R}^{N\times N}$, whilst $O_{N}$ denotes the zero matrix. For any $A, B\in\mathbb{R}^{N\times N}$, $A:B$ denotes the Frobenius inner product of $A$ and $B$.
\section{A New Lie Algebraic-Proof in the Case of Hard Spheres}\label{hardsphereproof}
Before we tackle the more complicated case of convex non-spherical particles in this paper in section \ref{newscatteringproof} below, to illustrate our new method, let us revisit the well-studied case of spherical particles which has been tackled by Boltzmann \cite{boltzmann2012wissenschaftliche}, Gr\"{o}nwall \cite{MR1503514}, and Cercignani \cite{MR1049048}, among others, under various regularity or summability hypotheses on the collision invariant. In all that follows, we assume without loss of generality that the mass of the hard spheres is normalised to unity, and that their radii are both $\frac{1}{2}$.

We recall once again that a continuously-differentiable map $\phi\in C^{1}(\mathbb{R}^{3})$ is a collision invariant for the Boltzmann equation for spherical particles if and only if it satisfies the functional equation
\begin{equation}\label{funkyspheres}
\phi(v_{n}')+\phi(\overline{v}_{n}')=\phi(v)+\phi(\overline{v}) 
\end{equation}
for all $v, \overline{v}\in\mathbb{R}^{3}$ and all $n\in\mathbb{S}^{2}$, where $v_{n}'\in\mathbb{R}^{3}$ and $\overline{v}_{n}'\in\mathbb{R}^{3}$ denote the $n$-dependent velocity vectors 
\begin{equation*}
\left\{
\begin{array}{l}
v_{n}':=v-((v-\overline{v})\cdot n)n, \vspace{2mm}\\
\overline{v}_{n}':=\overline{v}+((v-\overline{v})\cdot n)n.
\end{array}
\right.
\end{equation*}
The basic idea underlying our new approach is that, subject to a sequence of transformations which we detail below in section \ref{lift} onwards, the problem of finding a map $\phi\in C^{1}(\mathbb{R}^{3})$ satisfying the functional equation \eqref{funkyspheres} is equivalent to finding a map $f\in C^{1}(\mathbb{R}^{3\times 3})$ satisfying
\begin{equation}\label{simps}
f(g)=f(I_{3})    
\end{equation}
for all $g\in G_{\ast}$, where $G_{\ast}\subseteq \mathrm{O}(3)$ is termed the \emph{reduced scattering group} associated to the physical scattering family $\mathscr{S}=\{\sigma_{n}\}_{n\in\mathbb{S}^{2}}\subset\mathrm{O}(6)$ for hard spheres. In other words, the problem of characterisation of collision invariants is essentially one of characterising functions which are constant on a group. 

Unlike in the article \cite{saint2018collision}, in which the collision invariant is assumed only to be measurable with respect to the Lebesgue sigma algebra on $\mathbb{R}^{3}$, we assume herein that $\phi\in C^{1}(\mathbb{R}^{3})$, which in turn leads to $f$ being a continuously-differentiable map. This regularity assumption permits us to use information from the Lie algebra associated to $G_{\ast}$ to infer structural information on $f$. Indeed, by differentiation at the identity $I\in G_{\ast}$, it follows from \eqref{simps} that
\begin{equation*}
A:Df(I_{3})=0    
\end{equation*}
for all $A\in\mathfrak{g}_{\ast}$, where $\mathfrak{g}_{\ast}$ denotes the Lie algebra associated to the reduced scattering group $G_{\ast}$. As it can be shown that $\mathfrak{g}_{\ast}=\mathfrak{so}(3)$, it follows that $Df(I_{3})\in\mathbb{R}^{3\times 3}$ is a symmetric matrix. As a consequence of the symmetry of $Df(I_{3})$, we claim it follows that 
\begin{equation*}
\phi(v)=\widetilde{\Phi}(v, |v|^{2})    
\end{equation*}
for some map $\widetilde{\Phi}:\mathbb{R}^{3}\times [0, \infty)\rightarrow\mathbb{R}$, from which it follows by a standard argument that 
\begin{equation*}
\phi(v)=a+b\cdot v+c|v|^{2}    
\end{equation*}
for some $a, c\in\mathbb{R}$ and $b\in\mathbb{R}^{3}$. To the knowledge of the author, the material in this section constitutes a novel proof of the classical result of Boltzmann \cite{boltzmann2012wissenschaftliche} in the case that the collision invariant is assumed to be of class $\mathscr{C}^{1}$. Let us now elaborate on the method outlined above.
\subsection{Lifting the Problem from $\mathbb{R}^{3}$ to $\mathbb{R}^{6}$}\label{lift}
As was first noted in \cite{saint2018collision}, it is possible to rewrite the functional identity \eqref{funkyspheres} so as to bring out the role of group theory in the characterisation of collision invariants. Indeed, let us assume there exists \emph{at least} one $\phi\in C^{1}(\mathbb{R}^{3})$ which satisfies \eqref{funkyspheres}. Now, we define a new $\phi$-dependent auxiliary function $\Phi_{\phi}:\mathbb{R}^{6}\rightarrow\mathbb{R}^{6}$ by
\begin{equation*}
\Phi_{\phi}(V):=\phi(v)+\phi(\overline{v})    
\end{equation*}
for all $V=[v, \overline{v}]\in\mathbb{R}^{6}$, where as above square brackets $[\cdot, \cdot]$ denote the natural concatenation operator on pairs of vectors in $\mathbb{R}^{3}$ which produces a vector in $\mathbb{R}^{6}$, i.e.
\begin{equation*}
[v, \overline{v}]:=(v_{1}, v_{2}, v_{3}, \overline{v}_{1}, \overline{v}_{2}, \overline{v}_{3})  
\end{equation*}
for $v=(v_{1}, v_{2}, v_{3})\in\mathbb{R}^{3}$ and $\overline{v}=(\overline{v}_{1}, \overline{v}_{2}, \overline{v}_{3})\in\mathbb{R}^{3}$. In turn, the concatenated vector $V_{n}':=[v_{n}', \overline{v}_{n}']$ may be written as
\begin{equation*}
V_{n}'=\sigma_{n}V \quad\text{in}\hspace{2mm}\mathbb{R}^{6},    
\end{equation*}
where $\sigma_{n}\in\mathrm{O}(6)$ denotes the reflection matrix given by
\begin{equation*}
\sigma_{n}:=I_{6}-2\widehat{\nu}_{n}\otimes \widehat{\nu}_{n}    
\end{equation*}
with $\widehat{\nu}_{n}\in\mathbb{S}^{5}$ denoting the unit vector \begin{equation}\label{bigennspherenormal}
\widehat{\nu}_{n}:=\frac{1}{\sqrt{2}}\left[ 
\begin{array}{c}
n \\ 
-n
\end{array}
\right].
\end{equation}
It can be proved easily that $\sigma_{n}:V\mapsto V_{n}'$ effects the unique map on $\mathbb{R}^{6}$, linear or otherwise, which conserves the total linear momentum, angular momentum, and kinetic energy of its argument, whilst also ensuring that the hard particles do not interpenetrate following a collision. In turn, we term the collection of matrices $\mathscr{S}:=\{\sigma_{n}\}_{n\in\mathbb{S}^{2}}$ the \emph{physical scattering family} associated to the dynamics of two hard spheres. 

It follows from the above that the new auxiliary function $\Phi_{\phi}$ satisfies the identity
\begin{equation}\label{newidone}
\Phi_{\phi}(\sigma_{n}V)=\Phi_{\phi}(V)    
\end{equation}
for all $V\in\mathbb{R}^{6}$ and all $n\in\mathbb{S}^{2}$. In turn, from identity \eqref{newidone}, we may infer by iteration that
\begin{equation*}
\Phi_{\phi}\left(\prod_{i=1}^{N}\sigma_{n(i)}V\right)=\Phi_{\phi}(V)    
\end{equation*}
for any integer $N\geq 1$ and any collection of unit vectors $\{n(i)\}_{i=1}^{N}\subset \mathbb{S}^{2}$. In turn, as $\Phi_{\phi}$ is of class $\mathscr{C}^{1}$ on $\mathbb{R}^{6}$ and therefore continuous thereon, it follows that the auxiliary function $\Phi_{\phi}$ satisfies the identity
\begin{equation}\label{bigphiid}
\Phi_{\phi}(g V)=\Phi_{\phi}(V)       
\end{equation}
for all $g\in G$, where $G\subset\mathrm{O}(6)$ denotes the matrix group generated by the family of physical scattering maps $\mathscr{S}:=\{\sigma_{n}\}_{n\in\mathbb{S}^{2}}$, namely
\begin{equation*}
G:=\langle \mathscr{S}\rangle.
\end{equation*}
Moving forward, we shall term the group $G$ the \textbf{scattering group} associated to the unique physical scattering family $\mathscr{S}$ for hard spheres.

By the work of Eaton and Perlman \cite{MR0463329}, owing to the fact that $G\subset\mathrm{O}(6)$ is generated by a smoothly-parametrised family of reflection matrices in $\mathrm{O}(6)$, the scattering group $G$ is a strict subset of $\mathrm{O}(6)$ that admits the structure of a Lie group. As such, by the above argument, the problem of characterisation of collision invariants $\phi$ is settled by the characterisation of functions on $\mathbb{R}^{6}$ which are invariant under the action of the Lie group $G$. However, this characterisation problem for $G$-invariant functions is made more straightforward by noting that it may be reduced to one for the characterisation of $O(3)$-invariant functions, as we detail now.
\subsection{Dimension Reduction}
An observation which makes the task of finding all maps $\Phi_{\phi}$ which satisfy \eqref{bigphiid} above much simpler is that this characterisation problem is essentially one concerning maps on the 2-sphere $\mathbb{S}^{2}$, \emph{not} on Euclidean 6-space $\mathbb{R}^{6}$. Indeed, given a momentum vector $p\in\mathbb{R}^{3}$ and a value for the kinetic energy $e^{2}>0$ belonging to the admissible set $\mathcal{A}$ defined by
\begin{equation*}
\mathcal{A}:=\left\{(e, p)\in (0, \infty)\times\mathbb{R}^{3}\,:\, e^{2}>\frac{1}{2}|p|^{2}\right\},    
\end{equation*}
the associated \emph{energy-momentum manifold} $\mathcal{M}(e, p)\subset\mathbb{R}^{6}$ is defined to be
\begin{equation*}
\mathcal{M}(p, e):=\left\{
V\in\mathbb{R}^{6}\,:\, |V|^{2}=e^{2}\hspace{2mm}\text{and}\hspace{2mm} \left( \begin{array}{c}
V_{1} + V_{4} \\
V_{2} + V_{5} \\
V_{3} + V_{6}
\end{array}
\right)=p
\right\}
\end{equation*}
It is readily verified that, for any $(e, p)\in\mathcal{A}$, the energy-momentum manifold $\mathcal{M}(e, p)$ is diffeomorphic to the 2-sphere $\mathbb{S}^{2}$. Explicitly, if we define the map $H:\mathcal{A}\times\mathbb{R}^{3}\setminus\{0\}\rightarrow\mathbb{R}^{6}\setminus\{0\}$ pointwise by
\begin{equation*}
H(e, p, y):=\frac{1}{2}\left[
\begin{array}{c}
p-\sqrt{2e^{2}-|p|^{2}}y \\
p+\sqrt{2e^{2}-|p|^{2}}y
\end{array}
\right]    
\end{equation*}
for $(e, p, y)\in\mathcal{A}\times\mathbb{R}^{3}\setminus\{0\}$, it can be checked that $H$ is smooth, that the restriction map $H|_{\mathcal{A}\times\mathbb{S}^{2}}:\mathcal{A}\times\mathbb{S}^{2}\rightarrow H|_{\mathcal{A}\times\mathbb{S}^{2}}(\mathcal{A}\times\mathbb{S}^{2})$ is a diffeomorphism, and that
\begin{equation*}
(H|_{\mathcal{A}\times\mathbb{S}^{2}})^{-1}(\mathcal{M}(e, p))=\{(e, p)\}\times\mathbb{S}^{2}.    
\end{equation*}
When we recast the collision invariant identity \eqref{bigphiid} in the coordinate system determined by the map $H|_{\mathcal{A}\times\mathbb{S}^{2}}$, we may reduce the dimension of the domain of the map to be characterised. To this aim, we define a new $\phi$-dependent map $\Psi_{\phi}:\mathcal{A}\times\mathbb{R}^{3}\setminus\{0\}\rightarrow\mathbb{R}$ by
\begin{equation*}
\Psi_{\phi}(e, p, y):=\Phi_{\phi}(H(e, p, y))    
\end{equation*}
for $(e, p, y)\in\mathcal{A}\times\mathbb{R}^{3}\setminus\{0\}$. A brief calculation reveals that
\begin{equation*}
\sigma_{n}H(e, p, y)=H(e, p, s_{n}y),    
\end{equation*}
for all $n, y\in\mathbb{S}^{2}$, where $s_{n}\in\mathrm{O}(3)$ is the $n$-dependent reflection matrix defined by
\begin{equation*}
s_{n}:=I_{3}-2n\otimes n.    
\end{equation*}
In turn, the identity \eqref{bigphiid} takes the form
\begin{equation*}
\Phi_{\phi}(\sigma_{n}H(e, p, m))=\Phi_{\phi}(H(e, p, m)) 
\end{equation*}
which, by definition of the auxiliary function $\Psi_{\phi}$, yields the identity
\begin{equation}\label{bigpsiid}
\Psi_{\phi}(e, p, s_{n}y)=\Psi_{\phi}(e, p, y)    
\end{equation}
for all $(e, p, y)\in\mathcal{A}\times\mathbb{S}^{2}$ and all $n\in\mathbb{S}^{2}$. By way of this identity, we note that the value of energy $e$ and momentum $p$ is essentially immaterial in the characterisation of collision invariants. As such, we shall continue to transform the identity \eqref{bigpsiid} so as to suppress the role of both energy and momentum. 

From a notational point of view, we suppress the role of arguments of a function by relegating them to subscript parameters thereof. Indeed, now defining the $\phi$- and $(e, p)$-dependent map $\Omega_{\phi, e, p}:\mathbb{R}^{3}\setminus\{0\}\rightarrow\mathbb{R}$ pointwise by
\begin{equation*}
\Omega_{\phi, e, p}(y):=\Psi_{\phi}(e, p, y)   
\end{equation*}
for $y\in\mathbb{R}^{3}\setminus\{0\}$, it follows from the identity \eqref{bigpsiid} that $\Omega_{\phi, e, p}$ satisfies
\begin{equation*}
\Omega_{\phi, e, p}(s_{n}y)=\Omega_{\phi, e, p}(y)    
\end{equation*} 
for all $n, y\in\mathbb{S}^{2}$. Mirroring the approach of the previous section, we infer from an iteration argument that
\begin{equation*}
\Omega_{\phi, e, p}\left(\prod_{i=1}^{N}s_{n(i)}y\right)=\Omega_{\phi, e, p}(y)    
\end{equation*}
for any integer $N\geq 1$ and any collection of unit vectors $\{n(i)\}_{i=1}^{N}\subset\mathbb{S}^{2}$. Thus, as $\Omega_{\phi, e, p}\in C^{1}(\mathbb{R}^{3}\setminus\{0\})$ owing to the fact that $\phi\in C^{1}(\mathbb{R}^{3})$, we infer that
\begin{equation}\label{bigomegaid}
\Omega_{\phi, e, p}(gy)=\Omega_{\phi, e, p}(y)    
\end{equation}
for all $y\in\mathbb{S}^{2}$ and for all elements $g\in G_{\ast}$, where $G_{\ast}\subset\mathrm{O}(3)$ denotes the \textbf{reduced scattering group} associated to the unique physical scattering family $\mathscr{S}$ for hard spheres given by
\begin{equation*}
G_{\ast}:=\langle \{s_{n}\,:\,n\in\mathbb{S}^{2}\}\rangle.  
\end{equation*}
The reader will note that, by virtue of mapping the energy-momentum manifold to the 2-sphere, we have reduced the dimensionality of the problem at hand, in the sense that we seek to characterise all $G_{\ast}$-invariant maps on $\mathbb{R}^{3}$, as opposed to all $G$-invariant maps on $\mathbb{R}^{6}$. Following the approach of \cite{saint2018collision}, the following lemma allows us to determine the reduced scattering group $G_{\ast}$ exactly.
\begin{lem}
The reduced scattering group $G_{\ast}$ is the orthogonal group $O(3)$.
\end{lem}
\begin{proof}
This follows from \cite{MR0463329} or (\cite{saint2018collision}, Appendix), due to the fact that the set of reflections $\{s_{n}\}_{n\in\mathbb{S}^{2}}$ is parametrised by the 2-sphere $\mathbb{S}^{2}$ whose real linear span is $\mathbb{R}^{3}$.
\end{proof}
As a direct consequence of this lemma, owing to the fact that $\mathrm{O}(3)$ acts transitively on $\mathbb{S}^{2}$, we infer that if $\Omega_{\phi, e, p}$ satisfies \eqref{bigomegaid}, then it is a constant map. In turn, there exists a $\phi$- and $(e, p)$-dependent constant $\widetilde{\Omega}_{\phi, e, p}\in\mathbb{R}$ such that
\begin{equation*}
\Omega_{\phi, e, p}(y)=\widetilde{\Omega}_{\phi, e, p}    
\end{equation*}
for all $y\in\mathbb{S}^{2}$. By translating this information back to the original dependent variable $\Phi_{\phi}$ in terms of which $\Omega_{\phi, e, p}$ is defined, it holds that
\begin{equation*}
\Phi_{\phi}(V)=\widetilde{\Phi}_{\phi}(v+\overline{v}, |V|^{2})    
\end{equation*}
for some function $\widetilde{\Phi}_{\phi}:\mathbb{R}^{3}\times [0, \infty)\rightarrow\mathbb{R}$ and for all $V\in\mathbb{R}^{6}$, whence
\begin{equation*}
\phi(v)=\widetilde{\Phi}_{\phi}(v, |v|^{2})    
\end{equation*}
for all $v=[v, \overline{v}]\in\mathbb{R}^{3}$. A standard argument (such as that contained in the book of Truesdell and Muncaster \cite{truesdell1980fundamentals}) finally yields that the collision invariant $\phi\in C^{1}(\mathbb{R}^{3})$ is necessarily of the form
\begin{equation*}
\phi(v)=a+b\cdot v+c|v|^{2}    
\end{equation*}
for some constants $a, c\in\mathbb{R}$ and some constant vector $b\in\mathbb{R}^{3}$.

The success of the above algebro-geometric method depends on one being able to determine the reduced scattering group $G_{\ast}$ explicitly. However, in the case of non-spherical strictly-convex particles, the determination of the associated reduced scattering group associated to a given physical scattering family is, in general, a non-trivial task. Indeed, one typically faces the task of determining the group
\begin{equation*}
\langle S \rangle\subset\mathrm{O}(q),    
\end{equation*}
where $S\subset\mathrm{O}(q)$ is a strict subset of the orthogonal group of some dimension $q\geq 2$. Thanks to the work of \cite{MR0463329} or (\cite{saint2018collision}, Appendix), this task is achievable when the set $S$ is a collection of reflection matrices whose generating vectors in $\mathbb{S}^{q-1}$ span Euclidean space $\mathbb{R}^{q}$. When $S$ does not admit this structure, determining $\langle S \rangle$ is a challenge.

In preparation for the proof of our main theorem, namely Theorem \ref{maintheorem}, let us suppose in the context of hard spheres (artificially, of course) that we \emph{cannot} determine the reduced scattering group $G_{\ast}$ explicitly, i.e. that we do not have access to the results of \cite{MR0463329} or (\cite{saint2018collision}, Appendix). Under the assumption that $\phi\in C^{1}(\mathbb{R}^{3})$, we shall now demonstrate how information on the Lie algebra $\mathfrak{g}_{\ast}$ associated to the reduced scattering Lie group $G_{\ast}$ makes it possible to obtain the necessary information on $\Omega_{\phi, e, p}$ required for the characterisation of collision invariants of class $\mathscr{C}^{1}$.
\subsection{From Functions of Velocity to Functions of Scattering Matrices}
It is now that we depart from the method employed in Saint-Raymond and Wilkinson \cite{saint2018collision}. It will be convenient to suppress the dependence of our auxiliary dependent variable $\Omega_{\phi, e, p}$ on its $\mathbb{S}^{2}$-argument, and consider a new dependent variable which is a function of elements of the reduced scattering group $G_{\ast}$ alone. Now, as we have shown that the map $\Omega_{\phi, e, p}$ satisfies
\begin{equation*}
\Omega_{\phi, e, p}(gy)=\Omega_{\phi, e, p}(y)    
\end{equation*}
for all $g\in G_{\ast}$ and $y\in\mathbb{S}^{2}$, by defining a new auxiliary $\phi$-, $(e, p)$- and $y$-dependent map $f_{\phi, e, p, y}:\mathbb{R}^{3\times 3}\rightarrow\mathbb{R}$ pointwise by
\begin{equation*}
f_{\phi, e, p, y}(A):=\Omega_{\phi, p, e}(Ay) 
\end{equation*}
for $A\in \mathbb{R}^{3\times 3}\setminus\{0\}$, we obtain from \eqref{bigomegaid} above that $f_{\phi, e, p, y}$ satisfies the identity
\begin{equation}\label{ultsphere}
f_{\phi, e, p, y}(g)=f_{\phi, e, p, y}(I_{3})    
\end{equation}
for all $g\in G_{\ast}$. As $\phi$ being of class $\mathscr{C}^{1}$ implies that $f_{\phi, e, p, y}\in C^{1}(G_{\ast})$, by taking derivatives across identity \eqref{ultsphere} with respect to $g$ at the identity $I_{3}\in G_{\ast}$, we find that
\begin{equation}\label{orthog}
A:Df_{\phi, e, p, y}(I_{3})=0    
\end{equation}
for all $A\in \mathfrak{g}_{\ast}$. The identity \eqref{orthog} is interpreted as an orthogonality relation in the inner product space $\mathbb{R}^{3\times 3}$ endowed with the Frobenius inner product, and yields information on the structure of $Df_{\phi, e, p, y}(I_{3})$.

We are able to obtain the information we require on $Df_{\phi, e, p, y}(I_{3})$ for the purposes of characterising $\phi$ satisfying \eqref{funkyspheres} \emph{without} determining $\mathfrak{g}_{\ast}$ fully, rather by determining only a `large' linear subspace thereof (although, in the case of hard spheres, we can indeed show that $\mathfrak{g}_{\ast}=\mathfrak{so}(3)$). Now, by definition of the reduced scattering matrices $s_{n}\in G_{\ast}$ for $n\in\mathbb{S}^{2}$, we have that $f_{\phi, e, p, y}$ satisfies
\begin{equation*}
f_{\phi, e, p, y}(s_{n(\theta_{1}, \theta_{2})})=f_{\phi, e, p, y}(I_{3})  
\end{equation*}
for all $\theta_{1}\in[0, 2\pi)$ and $\theta_{2}\in [0, \pi)$, where
\begin{equation*}
n(\theta_{1}, \theta_{2}):=\left(
\begin{array}{c}
\cos\theta_{1}\sin\theta_{2} \\
\sin\theta_{1}\sin\theta_{2} \\
\cos\theta_{2}
\end{array}
\right).
\end{equation*}
In turn, one may show that
\begin{equation*}
(\partial_{\theta_{j}}n(\theta_{1}, \theta_{2})\otimes n(\theta_{1}, \theta_{2})-n(\theta_{1}, \theta_{2})\otimes \partial_{\theta_{j}}n(\theta_{1}, \theta_{2})):Df_{\phi, e, p, y}(I_{3})=0   
\end{equation*}
for $j\in\{1, 2\}$. By choosing the angles $\theta_{1}$ and $\theta_{2}$ appropriately, we find that
\begin{equation}\label{preliesphere}
A_{i}:Df_{\phi, e, p, m}(I_{3})=0    
\end{equation}
for $i\in\{1, 2, 3\}$, where
\begin{equation*}
A_{1}:=e_{1}\otimes e_{2}-e_{2}\otimes e_{1},
\end{equation*}
as well as
\begin{equation*}
A_{2}:=e_{1}\otimes e_{3}-e_{3}\otimes e_{1}    
\end{equation*}
and 
\begin{equation*}
A_{3}:=e_{2}\otimes e_{3}-e_{3}\otimes e_{2}.    
\end{equation*}
As $\{A_{1}, A_{2}, A_{3}\}\subset\mathbb{R}^{3\times 3}$ constitutes a basis for the vector space $\mathfrak{so}(3)$, from identity \eqref{preliesphere} we may deduce by linearity that
\begin{equation*}
A:Df_{\phi, e, p, y}(I_{3})=0    
\end{equation*}
for all $A\in \mathfrak{so}(3)$. Now, it follows from this orthogonality statement that $Df_{\phi, p, e, y}(I_{3})\in\mathbb{R}^{3\times 3}$ is necessarily a \emph{symmetric matrix}. It is this information on the derivative of $f_{\phi, e, p, y}$ that allows us, in essence, to characterise collision invariants $\phi\in C^{1}(\mathbb{R}^{3})$.

Indeed, by definition of the map $f_{\phi, e, p, y}$, we find by direct computation that
\begin{equation*}
Df_{\phi, e, p, y}(I_{3})=\nabla\Omega_{\phi, e, p}(y)\otimes y    
\end{equation*}
for all $y\in\mathbb{R}^{3}\setminus\{0\}$. Since we have deduced that $Df_{\phi, e, p, y}(I_{3})$ is necessarily symmetric, it holds by necessity that $\nabla\Omega_{\phi, e, p}$ satisfies the identity
\begin{equation}\label{rankone}
\nabla\Omega_{\phi, e, p}(y)\otimes y =y\otimes \nabla\Omega_{\phi, e, p}(y)       
\end{equation}
for all $y\in\mathbb{R}^{3}\setminus\{0\}$. At this point, an elementary PDE argument using the theory of classical solutions of initial-value problems of transport equations allows us to conclude that $\Omega_{\phi, e, p}$ is a constant map on the 2-sphere. Indeed, we apply the following elementary lemma. 
\begin{lem}
Suppose that a map $\Omega\in C^{1}(\mathbb{R}^{3}\setminus\{0\})$ satisfies the identity
\begin{equation}\label{nosubscripts}
\nabla\Omega(y)\otimes y =y\otimes \nabla\Omega(y)        
\end{equation}
for all $y\in\mathbb{R}^{3}$. It follows that $\Omega$ is a radially-symmetric function on $\mathbb{R}^{3}$.
\end{lem}
\begin{proof}
See Pinchover and Rubenstein (\cite{pinchover2005introduction}, Theorem 2.10).
\end{proof}
As a consequence of this lemma, there exists a $\phi$- and $(e, p )$-dependent constant $\widetilde{\Omega}_{\phi, e, p}\in\mathbb{R}$ such that
\begin{equation*}
\Omega_{\phi, e, p}(y)=\widetilde{\Omega}_{\phi, e, p}    
\end{equation*}
for all $y\in\mathbb{S}^{2}$. Thus, we have been able to show that $\Omega_{\phi, e, p}$ does not depend on its $\mathbb{S}^{2}$-argument in the absence of being able to determine the reduced scattering group $G_{\ast}$ exactly, which is the crucial step in leading to the characterisation result. To summarise, as a result of the above, we have proved the following statement which a new proof of the result originally obtained by Boltzmann.
\begin{prop}[Characterisation of $\mathscr{C}^{1}$ Collision Invariants]
If a map $\phi\in C^{1}(\mathbb{R}^{3})$ satisfies the identity
\begin{equation*}
\phi(v_{n}')+\phi(\overline{v}_{n}')=\phi(v)+\phi(\overline{v}) 
\end{equation*}
for all $n\in\mathbb{S}^{2}$ and all $v, \overline{v}\in\mathbb{R}^{3}$, then it is necessarily of the shape
\begin{equation*}
\phi(v)=a+b\cdot v+c|v|^{2}    
\end{equation*}
for all $v\in\mathbb{R}^{3}$, for some constants $a, c\in\mathbb{R}$ and some vector $b\in\mathbb{R}^{3}$.
\end{prop}
We now look to employ a similar approach to the above for our proof of Theorem \ref{maintheorem}. However, the proof in the case of non-spherical particles is more involved due, in part, to the fact that the scattering groups in this case \emph{depend on the orientations} of the underlying hard particles, which is a feature not present in the analysis of classical hard spheres.
\section{Proof for Non-Canonical Scattering for General Convex Particles}\label{newscatteringproof}
In this section, we prove Theorem \ref{maintheorem} above, namely the characterisation of all continuously-differentiable maps $\varphi:\mathbb{R}^{2}\times\mathbb{R}\times\mathbb{S}^{1}\rightarrow\mathbb{R}$ satisfying the functional equation
\begin{equation}\label{funkyfunky}
\varphi(v_{\beta}', \omega_{\beta}', \vartheta)+\varphi(\overline{v}_{\beta}', \overline{\omega}_{\beta}, \overline{\vartheta})=\varphi(v, \omega, \vartheta)+\varphi(\overline{v}, \overline{\omega}, \overline{\vartheta})  
\end{equation}
for all $v, \overline{v}\in\mathbb{R}^{2}$ and all $\omega, \overline{\omega}\in\mathbb{R}$, where $V_{\beta}'=\sigma_{\beta}^{\times}[V]$ is given in terms of the non-canonical scattering map $\sigma_{\beta}^{\times}$ by
\begin{equation*}
\sigma_{\beta}^{\times}[V]=M^{-1}(2\widehat{E}_{1}\otimes\widehat{E}_{1}+2\widehat{E}_{2}\otimes\widehat{E}_{2}+2\widehat{E}_{\beta}\otimes\widehat{E}_{\beta}-I_{6})MV    
\end{equation*}
for all $V\in\mathbb{R}^{6}$. We prove that any $\mathscr{S}_{\times}$-collision invariant is necessarily of the form
\begin{equation*}
\varphi(v, \omega, \vartheta)=a(\vartheta)+b\cdot v+c(m|v|^{2}+J\omega^{2})    
\end{equation*}
for some $a\in C^{1}(\mathbb{S}^{1})$, and some constants $b\in\mathbb{R}^{2}$ and $c\in\mathbb{R}$. The approach we take in what follows mirrors that employed in section \ref{hardsphereproof} above, in that we transform the identity \eqref{funkyfunky} repeatedly in order to reveal the role of the scattering group of lowest dimension pertinent to the non-canonical scattering family $\mathscr{S}_{\times}$ under consideration.
\subsection{Dimension Reduction}
As we shall set out, the problem of characterising $\mathscr{S}_{\times}$-collision invariants is essentially one of characterising $\mathrm{O}(3)$-invariant functions on $\mathbb{R}^{3}$. To see this, we shall transform coordinates. However, in the case of $\mathscr{S}_{\times}$-collision invariants, we require a sequence of \emph{five} transformations.

Let us begin by assuming that there exists an $\mathscr{S}_{\times}$-collision invariant, namely a map $\varphi\in C^{1}(\mathbb{R}^{2}\times\mathbb{R}\times\mathbb{S}^{1})$ satisfying the identity
\begin{equation}\label{collinv}
\varphi(v_{\beta}', \omega_{\beta}', \vartheta)+\varphi(\overline{v}_{\beta}', \overline{\omega}_{\beta}', \overline{\vartheta})=\varphi(v, \omega, \vartheta)+\varphi(\overline{v}, \overline{\omega}, \overline{\vartheta})    
\end{equation}
for all $v, \overline{v}\in\mathbb{R}^{2}$, $\omega, \overline{\omega}\in\mathbb{R}$, and $\beta=(\psi, \vartheta, \overline{\vartheta})\in\mathbb{T}^{3}$, where the scattering variables $v_{\beta}', \overline{v}_{\beta}'\in\mathbb{R}^{2}$ and $\omega_{\beta}', \overline{\omega}_{\beta}'\in\mathbb{R}$ are given by $V_{\beta}'=\sigma_{\beta}^{\times}[V]$. It will be of use to set some notation before we proceed. We set a new $\varphi$-dependent auxiliary map $\Phi_{\varphi}\in C^{1}(\mathbb{R}^{6}\times\mathbb{T}^{2})$ to be
\begin{equation*}
\Phi_{\varphi}(V, \vartheta, \overline{\vartheta}):=\varphi(v, \omega, \vartheta)+\varphi(\overline{v}, \overline{\omega}, \overline{\vartheta})
\end{equation*}
for $V\in\mathbb{R}^{6}$ and $(\vartheta, \overline{\vartheta})\in\mathbb{T}^{2}$. Using this notation, it can be checked that the collision invariant identity \eqref{collinv} is redressed as
\begin{equation}\label{phiofuse}
\Phi_{\varphi}(\sigma_{\beta}V, \vartheta, \overline{\vartheta})=\Phi_{\varphi}(V, \vartheta, \overline{\vartheta})
\end{equation}
for all $V\in\mathbb{R}^{6}$ and all $\beta=(\psi, \vartheta, \overline{\vartheta})\in\mathbb{T}^{3}$. If the reference particle $\mathsf{P}_{\ast}$ is a compact, convex set with boundary curve of class $\mathscr{C}^{1}$, we define set of all admissible energy-momentum pairs $\mathcal{A}_{\mathsf{P}_{\ast}}\subset (0, \infty)\times\mathbb{R}^{2}$ by
\begin{equation*}
\mathcal{A}_{\mathsf{P}_{\ast}}:=\left\{
(e, p)\in (0, \infty)\times \mathbb{R}^{2}\,:\,e^{2}>\frac{|p|^{2}}{2m}
\right\}.
\end{equation*}
For a given admissible energy-momentum pair $(e, p)\in\mathcal{A}_{\mathsf{P}_{\ast}}$, the associated energy-momentum manifold $\mathcal{M}_{\mathsf{P}_{\ast}}(e, p)\subset\mathbb{R}^{6}$ is defined to be 
\begin{equation*}
\mathcal{M}_{\mathsf{P}_{\ast}}(e, p):=\left\{
V\in\mathbb{R}^{6}\,:\,|MV|^{2}=e^{2} \hspace{2mm}\text{and}\hspace{2mm}m\left(
\begin{array}{c}
V_{1}+V_{3} \\
V_{2}+V_{4}
\end{array}
\right)=p
\right\},
\end{equation*}
where $M\in\mathbb{R}^{6\times 6}$ is the energy-momentum matrix given by
\begin{equation*}
M:=\mathrm{diag}(\sqrt{m}, \sqrt{m}, \sqrt{m}, \sqrt{m}, \sqrt{J}, \sqrt{J}).    
\end{equation*}
We also define the map $H_{\mathsf{P}_{\ast}}:\mathcal{A}_{\mathsf{P}_{\ast}}\times \mathbb{R}^{4}\setminus\{0\}\rightarrow\mathbb{R}^{6}\setminus\{0\}$ pointwise by
\begin{equation*}
H_{\mathsf{P}_{\ast}}(e, p, y):=\left(
\begin{array}{c}
\frac{1}{2m}(p_{1}+\sqrt{2me^{2}-|p|^{2}}y_{1}) \vspace{2mm} \\
\frac{1}{2m}(p_{2}+\sqrt{2me^{2}-|p|^{2})}y_{2}) \vspace{2mm} \\
\frac{1}{2m}(p_{1}-\sqrt{2me^{2}-|p|^{2}}y_{1}) \vspace{2mm} \\
\frac{1}{2m}(p_{2}-\sqrt{2me^{2}-|p|^{2})}y_{2}) \vspace{2mm} \\
\frac{1}{\sqrt{2mJ}}\sqrt{2me^{2}-|p|^{2}}y_{3} \vspace{2mm} \\
\frac{1}{\sqrt{2mJ}}\sqrt{2me^{2}-|p|^{2}}y_{4}
\end{array}
\right)
\end{equation*}
for all $(e, p, y)\in \mathcal{A}_{\mathsf{P}_{\ast}}\times\mathbb{R}^{4}\setminus\{0\}$. It may be checked that $h_{\mathsf{P}_{\ast}}$ is continuously differentiable, that the restriction map $H_{\mathsf{P}_{\ast}}|_{\mathcal{A}_{\mathsf{P}_{\ast}}\times\mathbb{S}^{3}}:\mathcal{A}_{\mathsf{P}_{\ast}}\times\mathbb{S}^{3}\rightarrow H_{\mathsf{P}_{\ast}}|_{\mathcal{A}_{\mathsf{P}_{\ast}}\times\mathbb{S}^{3}}(\mathcal{A}_{\mathsf{P}_{\ast}}\times\mathbb{S}^{3})$ admits the structure of a $\mathscr{C}^{1}$-diffeomorphism, and in turn that
\begin{equation*}
(H_{\mathsf{P}_{\ast}}|_{\mathcal{A}_{\mathsf{P}_{\ast}\times\mathbb{S}^{3}}})^{-1}(\mathcal{M}_{\mathsf{P}_{\ast}}(e, p))=\{(e, p)\}\times\mathbb{S}^{3}.    
\end{equation*}
With these definitions in place, we approach the following proposition.
\begin{prop}\label{usemeprop}
Let $\mathsf{P}_{\ast}\subset\mathbb{R}^{2}$ be a compact, convex set with boundary curve of class $\mathscr{C}^{1}$. Suppose that a map $\varphi\in C^{1}(\mathbb{R}^{2}\times\mathbb{R}\times\mathbb{S}^{1})$ is an $\mathscr{S}_{\times}$-collision invariant. For any $(e, p)\in\mathcal{A}_{\mathsf{P}_{\ast}}$ and any $y\in\mathbb{S}^{3}$, it follows that the map $f_{\varphi, e, p, y}:\mathbb{R}^{4\times 4}\times\mathbb{T}^{2}\rightarrow\mathbb{R}$ defined pointwise by
\begin{equation*}
f_{\varphi, e, p, y}(A, \vartheta, \overline{\vartheta}):=\Phi_{\varphi}(H_{\mathsf{P}_{\ast}}(e, p, Ay), \vartheta, \overline{\vartheta})    
\end{equation*}
for $A\in\mathbb{R}^{4\times 4}$ is of class $\mathscr{C}^{1}$ and satisfies the identity
\begin{equation*}
f_{\varphi, e, p, y}(g, \vartheta, \overline{\vartheta})=f_{\varphi, e, p, y}(I_{4}, \vartheta, \overline{\vartheta})
\end{equation*}
for all $g\in G_{\ast}(\vartheta, \overline{\vartheta})$, where $G_{\ast}\subset\mathrm{O}(4)$ denotes the reduced scattering group given by
\begin{equation*}
G_{\ast}(\vartheta, \overline{\vartheta}):=\langle\{s_{\beta}\}_{\psi\in\mathbb{S}^{1}}\rangle
\end{equation*}
with $s_{\beta}:=I_{4}-2\widehat{k}_{\beta}\otimes\widehat{k}_{\beta}\in\mathrm{O}(4)$ the reflection matrix generated by the unit vector $\widehat{k}_{\beta}\in\mathbb{S}^{4}$ defined by
\begin{equation}\label{knorm}
\widehat{k}_{\beta}:=\frac{1}{\sqrt{md_{\beta}^{2}+4J}}\left(
\begin{array}{c}
-\sqrt{m}d_{\beta}\sin\psi \vspace{1mm} \\
\sqrt{m}d_{\beta}\cos\psi \vspace{1mm} \\
-\sqrt{2J} \vspace{1mm} \\
-\sqrt{2J}
\end{array}
\right) 
\end{equation}
for $\beta=(\psi, \vartheta, \overline{\vartheta})\in\mathbb{T}^{3}$.
\end{prop}
\begin{proof}
We follow an approach similar to the case of hard spheres outlined in section \ref{hardsphereproof} above. Transforming the identity \eqref{phiofuse} in the auxiliary map $\Phi_{\varphi}$ to coordinates on the energy-momentum manifold, we find that the new auxiliary map $\Psi_{\varphi}:\mathcal{A}_{\mathsf{P}_{\ast}}\times\mathbb{R}^{4}\setminus\{0\}\times\mathbb{T}^{2}\rightarrow\mathbb{R}$ defined pointwise by
\begin{equation*}
\Psi_{\varphi}(e, p, y, \vartheta, \overline{\vartheta}):=\Phi_{\varphi}(H_{\mathsf{P}_{\ast}}(e, p, y), \vartheta, \overline{\vartheta})    
\end{equation*}
for $(e, p)\in\mathscr{A}_{\mathsf{P}_{\ast}}$, $y\in\mathbb{R}^{4}\setminus\{0\}$ and $(\vartheta, \overline{\vartheta})\in\mathbb{T}^{2}$ satisfies the identity
\begin{equation}\label{nodep}
\Psi_{\varphi}(e, p, s_{\beta}y, \vartheta, \overline{\vartheta})=\Psi_{\varphi}(e, p, y, \vartheta, \overline{\vartheta})   
\end{equation}
owing to the fact that
\begin{equation*}
\sigma_{\beta}^{\times}H_{\mathsf{P}_{\ast}}(e, p, y, \vartheta, \overline{\vartheta})=H_{\mathsf{P}_{\ast}}(e, p, s_{\beta}y, \vartheta, \overline{\vartheta})    
\end{equation*}
for all $\beta\in\mathbb{T}^{3}$, where $s_{\beta}\in\mathrm{O}(4)$ denotes the reflection matrix defined by
\begin{equation*}
s_{\beta}:=I-2\widehat{k}_{\beta}\otimes \widehat{k}_{\beta},
\end{equation*}
and $\widehat{k}_{\beta}\in\mathbb{S}^{3}$ denotes the unit vector \eqref{knorm} above. Now, by iteration, identity \eqref{nodep} reveals that if $\varphi$ is an $\mathscr{S}_{\times}$-collision invariant, then the auxiliary map $\Psi_{\varphi}$ satisfies 
\begin{equation}\label{redscat}
\Psi_{\varphi}(e, p, gy, \vartheta, \overline{\vartheta})=\Psi_{\varphi}(e, p, y, \vartheta, \overline{\vartheta})    
\end{equation}
for all $g\in G_{\ast}(\vartheta, \overline{\vartheta})$, where $G_{\ast}(\vartheta, \overline{\vartheta})$ denotes the potentially $(\vartheta, \overline{\vartheta})$-dependent \emph{reduced scattering group} given by
\begin{equation*}
G_{\ast}(\vartheta, \overline{\vartheta}):=\langle \{s_{\beta}\}_{\psi\in\mathbb{S}^{1}}\rangle.    
\end{equation*}
We note that the reduced scattering group $G_{\ast}(\vartheta, \overline{\vartheta})\subseteq\mathrm{O}(4)$ may depend on the orientation variables $(\vartheta, \overline{\vartheta})\in\mathbb{T}^{2}$ as the dependence of the auxiliary map $\Psi_{\varphi}$ on its $(\vartheta, \overline{\vartheta})$ argument, inherited from the dependence of the collision invariant $\varphi$ on its orientation variable, means that iteration from identity \eqref{nodep} to identity \eqref{redscat} can only take place with respect to the argument $\psi$. It is evident that identity \eqref{redscat} reveals that the values of both energy $e>0$ and momentum $p\in\mathbb{R}^{2}$ is immaterial to the characterisation of the collision invariant $\varphi$. 

Next, by defining a new $\varphi$-, $e$- and $p$-parametrised map $\Omega_{\varphi, e, p}:\mathbb{R}^{4}\setminus\{0\}\times\mathbb{T}^{2}\rightarrow\mathbb{R}$ pointwise by
\begin{equation*}
\Omega_{\varphi, e, p}(y, \vartheta, \overline{\vartheta}):=\Psi_{\varphi}(e, p, y, \vartheta, \overline{\vartheta}) 
\end{equation*}
for all $y\in\mathbb{R}^{4}\setminus\{0\}$ and all $(\vartheta, \overline{\vartheta})\in\mathbb{T}^{2}$, the collision invariant identity \eqref{nodep} now takes the form
\begin{equation}\label{omfy}
\Omega_{\varphi, e, p}(s_{\beta}y, \vartheta, \overline{\vartheta})=\Omega_{\varphi, e, p}(y, \vartheta, \overline{\vartheta})
\end{equation}
for all $y\in\mathbb{S}^{3}$ and all $\beta\in\mathbb{T}^{3}$ which, by way of iteration with respect to the collision parameter $\psi\in\mathbb{S}^{1}$, yields the identity
\begin{equation*}
    \Omega_{\varphi, e, p}(gy, \vartheta, \overline{\vartheta})=\Omega_{\varphi, e, p}(y, \vartheta, \overline{\vartheta})
\end{equation*}
for all $g\in G_{\ast}(\vartheta, \overline{\vartheta})$ and $y\in\mathbb{S}^{3}$. Finally, we define a $\varphi$-, $e$-, $p$- and $y$-parametrised map $f_{\varphi, e, p, y}:\mathbb{R}^{4\times 4}\times\mathbb{T}^{2}\rightarrow\mathbb{R}$ pointwise by
\begin{equation*}
f_{\varphi, e, p, y}(A, \vartheta, \overline{\vartheta}):=\Omega_{\varphi, e, p}(Ay, \vartheta, \overline{\vartheta})    
\end{equation*}
for all $A\in\mathbb{R}^{4\times 4}$, which from \eqref{omfy} reveals that $f_{\varphi, e, p, y}$ satisfies the identity
\begin{equation*}
f_{\varphi, e, p, y}(g, \vartheta, \overline{\vartheta})=f_{\varphi, e, p, y}(I_{4}, \vartheta, \overline{\vartheta})    
\end{equation*}
for all $g\in G_{\ast}(\vartheta, \overline{\vartheta})$, which is the result claimed in the statement of the proposition.
\end{proof}
Let us now remark on a notable difference between the case of collision invariants for non-canonical scattering $\mathscr{S}_{\times}$ and the case of hard spheres, as well as the case of canonical scattering $\mathscr{S}_{\ast}$ for non-spherical particles. The result of Eaton and Perlman \cite{MR0463329} and Viterbo (\cite{saint2018collision}, Appendix) does \emph{not} apply at this stage, since the linear span of the range of the map $\beta\mapsto \widehat{k}_{\beta}$ is the strict linear subspace of $\mathbb{R}^{4}$ given by
\begin{equation*}
\left\{
\left(
\begin{array}{c}
a \\ b \\ c \\ c
\end{array}
\right)\in\mathbb{R}^{4}\,:\,a, b, c\in\mathbb{R}
\right\}.
\end{equation*}
In order to make progress in the case of the non-canonical physical scattering family $\mathscr{S}_{\times}$, we must reduce the dimension of the problem further still.
\subsection{A Further Reduction of Dimension in the Case of $\mathscr{S}_{\times}$}\label{furreddy}
It is prudent to reduce the dimensionality of our problem yet further in the case of the non-canonical physical scattering family $\mathscr{S}_{\times}$. To see this, we proceed by investigating the properties of $Df_{\varphi, e, p, m}$ at the identity matrix. Indeed, by Proposition \ref{usemeprop} above, it holds that
\begin{equation*}
f_{\varphi, e, p, y}(s_{\beta_{1}}s_{\beta_{2}}, \vartheta, \overline{\vartheta})=f_{\varphi, e, p, y}(I_{4}, \vartheta, \overline{\vartheta})
\end{equation*}
for all $\beta_{1}, \beta_{2}\in\mathbb{T}^{3}$ of the form
\begin{equation*}
\beta_{i}=(\psi_{i}, \vartheta, \overline{\vartheta}),
\end{equation*}
where $\psi_{i}\in\mathbb{S}^{1}$ for $i\in\{1, 2\}$, which leads by differentiation of $f_{\varphi, e, p, y}(\cdot, \vartheta, \overline{\vartheta})$ at the identity to
\begin{equation}\label{needmoreinfo}
\left(\frac{\partial}{\partial\psi}\widehat{k}_{\beta}\otimes \widehat{k}_{\beta}-\widehat{k}_{\beta}\otimes \frac{\partial}{\partial \psi}\widehat{k}_{\beta}\right):Df_{\varphi, e, p, y}(I_{4}, \vartheta, \overline{\vartheta})   
\end{equation}
for all $\psi\in\mathbb{S}^{1}$. A calculation reveals that the set of matrices $K(\vartheta, \overline{\vartheta})\subset\mathbb{R}^{4\times 4}$ defined by
\begin{equation}\label{spanmatrix}
K(\vartheta, \overline{\vartheta}):=\mathrm{span}\left\{
\frac{\partial}{\partial\psi}\widehat{k}_{\beta}\otimes \widehat{k}_{\beta}-\widehat{k}_{\beta}\otimes \frac{\partial}{\partial \psi}\widehat{k}_{\beta}\in\mathbb{R}^{4\times 4}\,:\,\psi\in\mathbb{S}^{1} \right\}
\end{equation}
is simply
\begin{equation*}
K(\vartheta, \overline{\vartheta})=\left\{\left( \begin{array}{cccc} 
0 & -a & -b & -b \\
a & 0 & -c & -c \\ 
b & c & 0 & 0 \\ 
b & c & 0 & 0
\end{array}\right)\in\mathbb{R}^{4\times 4}\,:\,a, b, c\in\mathbb{R}\right\}.
\end{equation*}
From identity \eqref{needmoreinfo}, we find that
\begin{equation*}
A:Df_{\varphi, e, p, y}(I_{4}, \vartheta, \overline{\vartheta})=0    
\end{equation*}
for all $A\in K(\vartheta, \overline{\vartheta})$. However, as $K(\vartheta, \overline{\vartheta})\neq \mathfrak{so}(4)$, we \emph{cannot} infer from this statement that $Df_{\varphi, e, p, m}(I_{4}, \vartheta, \overline{\vartheta})$ is a symmetric matrix. Nevertheless, guided by the observation that $K(\vartheta, \overline{\vartheta})$ is isomorphic to $\mathfrak{so}(3)$, a further change of coordinates allows us to place this problem in a framework which is, in essence, identical to that of the case of hard spheres.

We introduce yet one more auxiliary map $\Gamma_{\varphi, e, p}:\mathbb{R}^{3}\setminus\{0\}\times\mathbb{T}^{2}\rightarrow\mathbb{R}$, parametrised by $\varphi$, $e$, and $p$, which is defined implicitly by the relation
\begin{equation*}
\Gamma_{\varphi, e, p}(I_{\ast}y, \vartheta, \overline{\vartheta}):=\Omega_{\varphi, e, p}(y, \vartheta, \overline{\vartheta}),
\end{equation*}
for all $y\in\mathbb{R}^{4}\setminus\{0\}$, where $I_{\ast}\in\mathbb{R}^{3\times 4}$ denotes the matrix
\begin{equation*}
I_{\ast}:=\left(
\begin{array}{cccc}
1 & 0 & 0 & 0 \\
0 & 1 & 0 & 0 \\
0 & 0 & 1 & 1
\end{array}
\right).
\end{equation*}
At this point, a calculation reveals that
\begin{equation}\label{smallcalc1}
\begin{array}{ll}
& \displaystyle \Omega_{\varphi, e, p}(s_{\beta}y, \vartheta, \overline{\vartheta}) \vspace{2mm} \\
= & \displaystyle \Gamma_{\varphi, e, p}(I_{\ast}s_{\beta}y, \vartheta, \overline{\vartheta})\vspace{2mm} \\
= & \displaystyle \Gamma_{\varphi, e, p}(\Delta^{-1}r_{\beta}\Delta I_{\ast}y, \vartheta, \overline{\vartheta}),
\end{array}
\end{equation}
where $\Delta\in\mathbb{R}^{3\times 3}$ denotes the matrix
\begin{equation*}
\Delta:=\left(
\begin{array}{ccc}
1 & 0 & 0 \\
0 & 1 & 0 \\
0 & 0 & \frac{1}{\sqrt{2}}
\end{array}
\right),
\end{equation*}
and $r_{\beta}\in\mathrm{O}(3)$ denotes the reflection matrix
\begin{equation*}
r_{\beta}:=I_{3}-2\widehat{\gamma}_{\beta}\otimes\widehat{\gamma}_{\beta}
\end{equation*}
with associated unit vector $\widehat{\gamma}_{\beta}\in\mathbb{S}^{2}$ given by
\begin{equation*}
\widehat{\gamma}_{\beta}:=\frac{1}{\sqrt{md_{\beta}^{2}+4J}}\left(
\begin{array}{c}
-\sqrt{m}d_{\beta}\sin\psi \vspace{1mm} \\
\sqrt{m}d_{\beta}\cos\psi \vspace{1mm} \\
-2\sqrt{J}
\end{array}
\right).
\end{equation*}
Using the result of \eqref{smallcalc1}, together with the identity \eqref{omfy}, we find that the auxiliary function $\Gamma_{\varphi, e, p}$ satisfies 
\begin{equation*}
\Gamma_{\varphi, e, p}(\Delta^{-1}r_{\beta}\Delta z, \vartheta, \overline{\vartheta})=\Gamma_{\varphi, e, p}(z, \vartheta, \overline{\vartheta})    
\end{equation*}
for all $z\in\mathbb{R}^{3}\setminus\{0\}$. Owing to the natural conjugation structure in the above identity, we find by iteration that
\begin{equation*}
\Gamma_{\varphi, e, p}(\Delta^{-1}g\Delta z, \vartheta, \overline{\vartheta})=\Gamma_{\varphi, e, p}(z, \vartheta, \overline{\vartheta})  
\end{equation*}
for all $g\in G_{0}(\vartheta, \overline{\vartheta})$, where $G_{0}(\vartheta, \overline{\vartheta})\subseteq\mathrm{O}(3)$ denotes the group
\begin{equation*}
G_{0}(\vartheta, \overline{\vartheta}):=\langle\{r_{\beta}\}_{\psi\in\mathbb{S}^{1}} \rangle.    
\end{equation*}
Finally, defining $\Lambda_{\varphi, e, p, y}:\mathbb{R}^{3\times 3}\times\mathbb{T}^{2}\rightarrow\mathbb{R}$ pointwise by
\begin{equation*}
\Lambda_{\varphi, e, p, y}(A, \vartheta, \overline{\vartheta}):=\Gamma_{\varphi, e, p}(\Delta^{-1}A\Delta I_{\ast}y, \vartheta, \overline{\vartheta}),    
\end{equation*}
we discover that 
\begin{equation*}
\Lambda_{\varphi, e, p, y}(g, \vartheta, \overline{\vartheta})=\Lambda_{\varphi, e, p, y}(I_{3}, \vartheta, \overline{\vartheta})   
\end{equation*}
for all $g\in G_{0}(\vartheta, \overline{\vartheta})$. Thus, we have discovered that the `minimal' scattering group in this problem is a matrix subgroup of $\mathrm{O}(3)$. To summarise the discussion of this section, we state the following proposition.
\begin{prop}\label{helpme}
Let $\mathsf{P}_{\ast}\subset\mathbb{R}^{2}$ be a compact, convex set with boundary curve of class $\mathscr{C}^{1}$. Suppose that a map $\varphi\in C^{1}(\mathbb{R}^{2}\times\mathbb{R}\times\mathbb{S}^{1})$ is an $\mathscr{S}_{\times}$-collision invariant. For any $(e, p)\in\mathcal{A}_{\mathsf{P}_{\ast}}$ and any $y\in\mathbb{S}^{3}$, it follows that the map $\Lambda_{\varphi, e, p, y}:\mathbb{R}^{3\times 3}\setminus\{0\}\times\mathbb{T}^{2}\rightarrow \mathbb{R}$ defined pointwise by
\begin{equation*}
\Lambda_{\varphi, e, p, y}(A, \vartheta, \overline{\vartheta}):=\Gamma_{\varphi, e, p}(\Delta^{-1}A\Delta I_{\ast}y, \vartheta, \overline{\vartheta})    
\end{equation*}
is of class $\mathscr{C}^{1}$ and satisfies the identity
\begin{equation*}
\Lambda_{\varphi, e, p, y}(g, \vartheta, \overline{\vartheta})=\Lambda_{\varphi, e, p, y}(I_{3}, \vartheta, \overline{\vartheta})    
\end{equation*}
for all $g\in G_{0}(\vartheta, \overline{\vartheta})$.
\end{prop}
At this point, we are ready to exploit information on the map $\Lambda_{\varphi, e, p, y}$ to infer information on the original auxiliary map $\Phi_{\varphi}$. We do this in the following section.
\subsection{A Symmetry Condition on $D\Lambda_{\varphi, e, p, y}(I_{3}, \vartheta, \overline{\vartheta})$}
Let us begin with the following lemma.
\begin{lem}\label{symmat}
Let $\mathsf{P}_{\ast}\subset\mathbb{R}^{2}$ be a compact, convex set with boundary curve of class $\mathscr{C}^{1}$. Suppose that $\varphi\in C^{1}(\mathbb{R}^{2}\times\mathbb{R}\times\mathbb{S}^{1})$ is an $\mathscr{S}_{\times}$-collision invariant. For any $(e, p)\in\mathcal{A}_{\mathsf{P}_{\ast}}$, $(\vartheta, \overline{\vartheta})\in\mathbb{T}^{2}$, and any $y\in\mathbb{S}^{3}$, it holds that $D\Lambda_{\varphi, e, p, y}(I_{3}, \vartheta, \overline{\vartheta})\in\mathbb{R}^{3\times 3}$ is a symmetric matrix. 
\end{lem}
\begin{proof}
From Proposition \ref{helpme} above, we have that
\begin{equation*}
\Lambda_{\varphi, e, p, y}(r_{\beta_{1}}r_{\beta_{2}}, \vartheta, \overline{\vartheta})=\Lambda_{\varphi, e, p, y}(I_{3}, \vartheta, \overline{\vartheta})
\end{equation*}
for all $\beta_{1}, \beta_{2}\in\mathbb{T}^{3}$ of the form $\beta_{i}=(\psi_{i}, \vartheta, \overline{\vartheta})$, where $\psi_{i}\in\mathbb{S}^{1}$ for $i\in\{1, 2\}$. In turn, in a manner similar to the case considered in section \ref{furreddy} above, we find that
\begin{equation}\label{lee}
\left(\frac{\partial}{\partial\psi}\widehat{\gamma}_{\beta}\otimes\widehat{\gamma}_{\beta}-\widehat{\gamma}_{\beta}\otimes\frac{\partial}{\partial\psi}\widehat{\gamma}_{\beta}\right):D\Lambda_{\varphi, e, p, y}(I_{3}, \vartheta, \overline{\vartheta})=0    
\end{equation}
for all $\psi\in\mathbb{S}^{1}$. By suitable choices of $\psi\in\mathbb{S}^{1}$, it is straightforward to show that
\begin{equation*}
\mathrm{span}\left\{
\frac{\partial}{\partial\psi}\widehat{\gamma}_{\beta}\otimes\widehat{\gamma}_{\beta}-\widehat{\gamma}_{\beta}\otimes\frac{\partial}{\partial\psi}\widehat{\gamma}_{\beta}\,:\,\psi\in\mathbb{S}^{1}
\right\}=\mathfrak{so}(3),
\end{equation*}
whence from \eqref{lee} above we conclude that $D\Lambda_{\varphi, e, p, y}(I_{3}, \vartheta, \overline{\vartheta})$ satisfies
\begin{equation*}
A: D\Lambda_{\varphi, e, p, y}(I_{3}, \vartheta, \overline{\vartheta})=0   
\end{equation*}
for all $A\in\mathfrak{so}(3)$. In turn, $D\Lambda_{\varphi, e, p, y}(I_{3}, \vartheta, \overline{\vartheta})\in\mathbb{R}^{3\times 3}$ is a symmetric matrix as claimed.
\end{proof}
\begin{rem}
It is also possible to derive the conclusion of the above lemma using the result of \cite{MR0463329} and (\cite{saint2018collision}, Appendix). In particular, it is straightforward to show that
\begin{equation*}
\mathrm{span}\{
\widehat{\gamma}_{\beta}\,:\,\psi\in\mathbb{S}^{1}
\}=\mathbb{R}^{3},
\end{equation*}
from which we conclude that $\mathfrak{g}_{0}=\mathfrak{so}(3)$. Indeed, aiming for a contradiction, suppose that it holds that
\begin{equation}\label{ortho}
\widehat{\gamma}_{\beta}\cdot \xi=0    
\end{equation}
for all $\psi\in\mathbb{S}^{1}$ for some $\xi\in\mathbb{R}^{3}\setminus\{0\}$. Identity \eqref{ortho} is then equivalent to the statement that the distance of closest approach $d_{\beta}$ associated to the particle $\mathsf{P}_{\ast}$ is given pointwise by
\begin{equation*}
d_{\beta}=2\xi_{3}\sqrt{\frac{J}{m(\xi_{1}^{2}+\xi_{2}^{2})}}\mathrm{sech}(\psi+\alpha)    
\end{equation*}
for all $\psi\in\mathbb{S}^{1}$, where $\alpha:=\arctan(\xi_{2}/\xi_{1})$, which is absurd.
\end{rem}
As a useful consequence of the above lemma, we are able to obtain the following structural information on the auxiliary map $\Omega_{\varphi, e, p}$.
\begin{cor}
Let $\mathsf{P}_{\ast}\subset\mathbb{R}^{2}$ be a compact, convex set with boundary curve of class $\mathscr{C}^{1}$. Suppose that $\varphi\in C^{1}(\mathbb{R}^{2}\times\mathbb{R}\times\mathbb{S}^{1})$ is an $\mathscr{S}_{\times}$-collision invariant. For any $(e, p)\in\mathcal{A}_{\mathsf{P}_{\ast}}$, it holds that $\Omega_{\varphi, e, p}:\mathbb{R}^{4}\setminus\{0\}\times\mathbb{T}^{2}$ admits the representation
\begin{equation*}
\Omega_{\varphi, e, p}(y, \vartheta, \overline{\vartheta})=\widetilde{\Omega}_{\varphi, e, p}(\sqrt{1+2y_{3}y_{4}}, \vartheta, \overline{\vartheta})   
\end{equation*}
for all $y\in\mathbb{S}^{3}$ for some $\varphi$- and $(e, p)$-parametrised function $\widetilde{\Omega}_{\varphi, e, p}:(0, \infty)\times\mathbb{T}^{2}\rightarrow\mathbb{R}$.
\end{cor}
\begin{proof}
By definition of the auxiliary map $\Lambda_{\varphi, e, p, y}$, it holds that
\begin{equation*}
D\Lambda_{\varphi, e, p, y}(I_{3}, \vartheta, \overline{\vartheta})=\nabla\Gamma_{\varphi, e, p}(I_{\ast}y, \vartheta, \overline{\vartheta})\otimes I_{\ast}y    
\end{equation*}
for all $y\in\mathbb{R}^{4}\setminus\{0\}$. Converting this statement to information about the auxiliary map $\Omega_{\varphi, e, p}$, we in turn find that
\begin{equation*}
D\Lambda_{\varphi, e, p, y}(I_{3}, \vartheta, \overline{\vartheta})=I_{\ast}^{(i)}\nabla\Omega_{\varphi, e, p}(y, \vartheta, \overline{\vartheta})\otimes I_{\ast}y,    
\end{equation*}
where, for each $i\in\{1, 2\}$, the matrix $I_{\ast}^{(i)}\in\mathbb{R}^{3\times 4}$ is given by
\begin{equation*}
I_{\ast}^{(1)}:=\left(
\begin{array}{cccc}
1 & 0 & 0 & 0 \\
0 & 1 & 0 & 0 \\
0 & 0 & 1 & 0
\end{array}
\right) \quad \text{and}\quad I_{\ast}^{(2)}:=\left(
\begin{array}{cccc}
1 & 0 & 0 & 0 \\
0 & 1 & 0 & 0 \\
0 & 0 & 0 & 1
\end{array}
\right).
\end{equation*}
By the result of lemma \ref{symmat} above, we conclude that the components of $\nabla\Omega_{\varphi, e, p}\in\mathbb{R}^{4}$ satisfy the identities
\begin{equation*}
y_{2}\frac{\partial\Omega_{\varphi, e, p}}{\partial y_{1}}(y, \vartheta, \overline{\vartheta})=y_{1}\frac{\partial\Omega_{\varphi, e, p}}{\partial y_{2}}(y, \vartheta, \overline{\vartheta}),
\end{equation*}
\begin{equation*}
y_{1}\frac{\partial\Omega_{\varphi, e, p}}{\partial y_{3}}(y, \vartheta, \overline{\vartheta})=(y_{3}+y_{4})\frac{\partial\Omega_{\varphi, e, p}}{\partial y_{1}}(y, \vartheta, \overline{\vartheta}),
\end{equation*}
\begin{equation*}
y_{2}\frac{\partial\Omega_{\varphi, e, p}}{\partial y_{3}}(y, \vartheta, \overline{\vartheta})=(y_{3}+y_{4})\frac{\partial\Omega_{\varphi, e, p}}{\partial y_{2}}(y, \vartheta, \overline{\vartheta}),
\end{equation*}
\begin{equation*}
y_{1}\frac{\partial\Omega_{\varphi, e, p}}{\partial y_{4}}(y, \vartheta, \overline{\vartheta})=(y_{3}+y_{4})\frac{\partial\Omega_{\varphi, e, p}}{\partial y_{1}}(y, \vartheta, \overline{\vartheta})
\end{equation*}
and
\begin{equation*}
y_{2}\frac{\partial\Omega_{\varphi, e, p}}{\partial y_{4}}(y, \vartheta, \overline{\vartheta})=(y_{3}+y_{4})\frac{\partial\Omega_{\varphi, e, p}}{\partial y_{2}}(y, \vartheta, \overline{\vartheta})
\end{equation*}
for all $y\in\mathbb{R}^{4}\setminus\{0\}$. By employing an argument similar to that employed in the proof of Proposition \ref{usemeprop} above, we infer that there exists a function $\widetilde{\Omega}_{\varphi, e, p}:(0, \infty)\times\mathbb{T}^{2}\rightarrow\mathbb{R}$ such that
\begin{equation*}
\Omega_{\varphi, e, p}(y, \vartheta, \overline{\vartheta})=\widetilde{\Omega}_{\varphi, e, p}((y_{1}^{2}+y_{2}^{2}+(y_{3}+y_{4})^{2})^{1/2}, \vartheta, \overline{\vartheta})    
\end{equation*}
for all $y\in\mathbb{R}^{4}\setminus\{0\}$. In particular, for all $y\in\mathbb{S}^{3}$, owing to the fact that
\begin{equation*}
  y_{1}^{2}+y_{2}^{2}+(y_{3}+y_{4})^{2}=1+2y_{3}y_{4},
\end{equation*}
the statement of the corollary follows.
\end{proof}
In contrast to the case of canonical scattering, the essential group at the heart of the collision invariant problem for non-canonical scattering is $O(3)$, rather than $O(4)$. Due to this loss of symmetry, one might imagine that any $\mathscr{S}_{\times}$-collision invariant $\varphi$ is less symmetric than its $\mathscr{S}_{\ast}$ counterpart. Remarkably, although the `minimal' scattering group is smaller, any $\mathscr{S}_{\times}$-collision invariant is also a $\mathscr{S}_{\ast}$-collision invariant. Indeed, we may now prove our main theorem, namely Theorem \ref{maintheorem}.
\begin{thm}
Let $\mathsf{P}_{\ast}\subset\mathbb{R}^{2}$ be a compact, convex set with boundary curve of class $\mathscr{C}^{1}$. Suppose that $\varphi\in C^{1}(\mathbb{R}^{2}\times\mathbb{R}\times\mathbb{S}^{1})$ is an $\mathscr{S}_{\times}$-collision invariant. It holds that $\varphi$ is necessarily of the form
\begin{equation*}
\varphi(v, \omega, \vartheta)=a(\vartheta)+b\cdot v+c(m|v|^{2}+J\omega^{2})    
\end{equation*}
for some $a\in C^{1}(\mathbb{S}^{1})$, $b\in\mathbb{R}^{2}$, and $c\in\mathbb{R}$.
\end{thm}
\begin{proof}
We may assume, without loss of generality, that the $\mathscr{S}_{\times}$-collision invariant $\varphi$ has the property 
\begin{equation*}
\varphi(0, 0, \vartheta)=0    
\end{equation*}
for all $\vartheta\in\mathbb{S}^{1}$. As we know from that the auxiliary map $\Omega_{\varphi, e, p}$ is of the shape
\begin{equation*}
\Omega_{\varphi, e, p}(y)=\widetilde{\Omega}_{\varphi, e, p}((1+2y_{3}y_{4})^{1/2}, \vartheta, \overline{\vartheta})    
\end{equation*}
for all $y\in\mathbb{S}^{3}$, from the definition of $\Omega_{\varphi, e, p}$ in terms of the original auxiliary map $\Phi_{\varphi}$, it can be verified readily that
\begin{equation*}
\Phi_{\varphi}(V)=\widetilde{\Omega}_{\varphi, |MV|, p(V)}\left(\left(1+\frac{2mJV_{5}V_{6}}{2m|V|^{2}-|p(V)|^{2}}\right)^{1/2}, \vartheta, \overline{\vartheta}\right)    
\end{equation*}
for all $V\in\mathbb{R}^{6}$, where
\begin{equation*}
p(V):=m\left(
\begin{array}{c}
V_{1}+V_{3} \\
V_{2}+V_{4}
\end{array}
\right).
\end{equation*}
In turn, we infer the existence of a map $\widetilde{\Phi}_{\varphi}: (0, \infty)\times \mathbb{R}^{2}\times\mathbb{R}\times\mathbb{T}^{2}\rightarrow\mathbb{R}$ such that
\begin{equation*}
\Phi_{\varphi}(V)=\widetilde{\Phi}_{\varphi}\left(|MV|^{2}, p(V), \left(1+\frac{2mJV_{5}V_{6}}{2m|MV|^{2}-|p(V)|^{2}}\right)^{1/2}, \vartheta, \overline{\vartheta}\right)       
\end{equation*}
for all $V\in\mathbb{R}^{6}$. Converting this to information in terms of the $\mathscr{S}_{\times}$-collision invariant $\varphi$, we obtain that $\varphi$ satisfies the identity
\begin{equation}\label{almostdone}
\varphi(v, \omega, \vartheta)+\varphi(\overline{v}, \overline{\omega}, \overline{\vartheta})=\widetilde{\Phi}_{\varphi}\left(|MV|^{2}, p(V), \left(1+\frac{2mJ\omega\overline{\omega}}{2m|MV|^{2}-|p(V)|^{2}}\right)^{1/2}, \vartheta, \overline{\vartheta}\right)
\end{equation}
for all $V=[v, \overline{v}, \omega, \overline{\omega}]\in\mathbb{R}^{6}$ and $(\vartheta, \overline{\vartheta})\in\mathbb{T}^{2}$. By setting $\overline{v}=0$ and $\overline{\omega}=0$, we find that
\begin{equation}\label{reallyalmostdone}
\varphi(v, \omega, \vartheta)=\widetilde{\Phi}_{\varphi}(m|v|^{2}+J\omega^{2}, mv, 1, \vartheta, \overline{\vartheta})    
\end{equation}
for all $v\in\mathbb{R}^{2}$, $\omega\in\mathbb{R}$ and all $(\vartheta, \overline{\vartheta})\in\mathbb{T}^{2}$, whence $\widetilde{\Phi}_{\varphi}$ does \emph{not} depend on its $\overline{\vartheta}$-argument. Similarly, by setting $v=0$ and $\omega=0$ in \eqref{almostdone}, we also find that $\widetilde{\Phi}_{\varphi}$ does not depend on its $\vartheta$-argument. Thus, from \eqref{reallyalmostdone}, it holds that there exists a map $\Phi_{\varphi}^{\ast}:[0, \infty)\times \mathbb{R}^{2}\rightarrow\mathbb{R}$ such that
\begin{equation*}
\varphi(v, \omega, \vartheta)=\Phi_{\varphi}^{\ast}(m|v|^{2}+J\omega^{2}, mv)    
\end{equation*}
for all $v\in\mathbb{R}^{2}$ and $\omega\in\mathbb{R}$. By a routine argument, following the work of Truesdell and Muncaster (\cite{truesdell1980fundamentals}, Chapter VI), we find that the $\mathscr{S}_{\times}$-collision invariant $\varphi$ is necessarily of the form
\begin{equation*}
\varphi(v, \omega, \vartheta)=b\cdot v+c(m|v|^{2}+J\omega^{2})    
\end{equation*}
for all $v\in\mathbb{R}^{2}$, $\omega\in\mathbb{R}$ and $\vartheta\in\mathbb{S}^{2}$ for some constant $b\in\mathbb{R}^{2}$ and constant $c\in\mathbb{R}$. The proof of the main theorem follows.
\end{proof}
\section{Concluding Remarks}\label{conclusion}
In this paper, we have shown that, in the case of two-dimensional convex particles with boundary that is suitably regular, $\mathscr{S}_{\times}$-collision invariants $\varphi$ are necessarily of the form
\begin{equation}\label{unitwo}
\varphi(v, \omega, \vartheta)=a(\vartheta)+b\cdot v+c(m|v|^{2}+J\omega^{2})    
\end{equation}
for all $v\in\mathbb{R}^{2}$, $\omega\in\mathbb{R}$, and $\vartheta\in\mathbb{S}^{1}$. As a consequence of Theorem \ref{maintheorem}, together with the known result of \cite{saint2018collision}, it follows that all collision invariants for physical scattering families (as defined by Definition \ref{physscatdeftwo} in this work) in the two-dimensional case are of the form \eqref{unitwo} above. However, the problem for three-dimensional convex particles is more complicated. As reported in \cite{wilkinson2020non}, there are uncountably-many physical scattering families in the case of three-dimensional convex particles. As such, characterisation of all possible $\mathscr{S}$-collision invariants is rather more involved, and therefore was not tackled in the present article. 

Let us now comment on the physical significance of characterising collision invariants for all possible physical scattering families.
\subsection{Universality}
As a result of the multiplicity of physical scattering families in the three-dimensional case, one might ask, rather naturally, if all collision invariants are of the same shape, no matter the underlying choice of the hard particle scattering one makes. As the general theory of the classical Boltzmann equation makes clear, collision invariants reveal the \emph{mesoscopic} physical properties of a dilute gas (and, in the context of hydrodynamic limits \cite{saint2009hydrodynamic} the \emph{macroscopic} properties thereof). As such, if all collision invariants for convex particle scattering happen to be the same, irrespective of the underlying physical scattering family, one might conclude that the `observable' properties of a dilute gas do not depend on the precise details of the underlying scattering processes, only on the fact that scattering conserves total momentum and energy of the particle system, together with the condition that the particles do not interpenetrate.

In this direction, suppose that a reference hard particle $\mathsf{P}_{\ast}\subset\mathbb{R}^{3}$ is given and fixed. We shall say that a map $\varphi$ of the form
\begin{equation}\label{universal}
\varphi(v, \omega, R):=a(R)+b\cdot v+c(m|v|^{2}+RJR^{T}\omega\cdot \omega)    
\end{equation}
is a \textbf{universal collision invariant} if and only if every $\mathscr{S}$-collision invariant is of the form \eqref{universal}, no matter the choice of physical scattering family $\mathscr{S}$ one fixes to define the scattering process. It is not \emph{a priori} obvious that the structure of a collision invariant should only depend, essentially, on the linear momentum and the kinetic energy of a gas particle. The case of arbitrary physical scattering families in three dimensions comprised only of \emph{linear} maps will be treated in future work.

As regards universality, the problem of existence of \emph{nonlinear} classical solutions $\sigma_{\beta}$ of the Jacobian equation
\begin{equation*}
\mathrm{det}\,D\sigma_{\beta}=-1    
\end{equation*}
on $\mathbb{R}^{12}$ subject to the algebraic constraints from the conservation laws \eqref{nonspherecolm}--\eqref{nonspherecoke}, together with the semi-algebraic constraint
\begin{equation*}
\sigma_{\beta}[V]\cdot N_{\beta}\geq 0    
\end{equation*}
whenever $V\cdot N_{\beta}\leq 0$ is also of physical and mathematical interest. We subsequently denote families $\{\sigma_{\beta}\}_{\beta}$ of nonlinear physical scattering maps by $\mathscr{N}$. Were they to exist, the study of non-trivial $\mathscr{N}$-collision invariants would require techniques which are very different to those employed in this work in the case of non-canonical scattering. Indeed, as physical scattering families $\mathscr{S}$ of linear maps may be studied using techniques from Lie theory of finite-dimensional matrix groups, the analysis of physical scattering families $\mathscr{N}$ of nonlinear maps would require more general techniques from the theory of differentiable groups.
\subsection*{Acknowledgements}
The author would like to thank Heiko Gimperlein for interesting and illuminating discussions related to the material in this work.

\medskip

\bibliographystyle{siam}
\bibliography{bib}

\end{document}